\documentclass[aps,prl,twocolumn,showpacs,superscriptaddress,longbibliography]{revtex4-2}
\usepackage{makeidx}
\usepackage{graphicx}
\usepackage{amsmath}
\usepackage{amsthm}
\usepackage{mathrsfs}
\usepackage{graphicx}
\usepackage{braket}
\usepackage[subrefformat=parens,labelformat=parens,caption=false]{subfig}
\usepackage{mleftright}
\usepackage{amsfonts}
\usepackage{dcolumn}
\usepackage{bm}
\usepackage{bbm}
\usepackage{color}
\usepackage[dvipsnames]{xcolor}
\usepackage{esint}
\usepackage{lineno}
\usepackage{verbatim}
\usepackage{cancel}
\usepackage{appendix}
\usepackage{apptools}

\usepackage[breaklinks,
            colorlinks,
            urlcolor=blue,
            linkcolor=blue,
            anchorcolor=blue,
            citecolor=blue]{hyperref}        
\usepackage{soul}
\usepackage{amssymb}

\definecolor{THc}{rgb}{0.2,0.9,0.5}
 
\newcommand{\change}[1]{{\color{black} #1}}

\newtheorem{theorem}{Theorem}

\usepackage[linesnumbered, ruled,vlined]{algorithm2e}

\def\app#1#2{%
  \mathrel{%
    \setbox0=\hbox{$#1\sim$}%
    \setbox2=\hbox{%
      \rlap{\hbox{$#1\propto$}}%
      \lower1.1\ht0\box0%
    }%
    \raise0.25\ht2\box2%
  }%
}

\ifx\proof\undefined
\newenvironment{proof}[1][\protect\proofname]{\par
	\normalfont\topsep6\p@\@plus6\p@\relax
	\trivlist
	\itemindent\parindent
	\item[\hskip\labelsep\scshape #1]\ignorespaces
}{%
	\endtrivlist\@endpefalse
}
\providecommand{\proofname}{Proof}
\fi

\DeclareMathOperator{\E}{\mathbb{E}}

\newcommand{\bigO}[1]{\mathcal{O}\mleft(#1\mright)}

\newcommand{\Tr}{\mathrm{Tr}}

\providecommand{\factname}{Fact}
\providecommand{\theoremname}{Theorem}
\providecommand{\claimname}{Claim}
\providecommand{\lemmaname}{Lemma}
\providecommand{\definitionname}{Definition}
\providecommand{\corollaryname}{Corollary}
\providecommand{\conjecturename}{Conjecture}

\newcommand\norm[1]{\left\lVert#1\right\rVert}

\newcommand{\sectionMain}[1]{
\let\oldaddcontentsline\addcontentsline
\renewcommand{\addcontentsline}[3]{}
\section{#1}
\let\addcontentsline\oldaddcontentsline
}

\graphicspath{{./images/},{./imagesAppendix/}}

\begin{document}

\title{Optimal Conversion from Classical to Quantum Randomness via Quantum Chaos}
\author{Wai-Keong Mok}
\affiliation{Institute for Quantum Information and Matter, California Institute of Technology, Pasadena, CA 91125, USA}
\author{Tobias Haug}
\affiliation{Quantum Research Centre, Technology Innovation Institute, Abu Dhabi, UAE}
\author{Adam L. Shaw}
\thanks{Present address: Department of Physics, Stanford University, Stanford, CA}
\affiliation{Institute for Quantum Information and Matter, California Institute of Technology, Pasadena, CA 91125, USA}

\author{Manuel Endres}
\affiliation{Institute for Quantum Information and Matter, California Institute of Technology, Pasadena, CA 91125, USA}
\author{John Preskill}
\affiliation{Institute for Quantum Information and Matter, California Institute of Technology, Pasadena, CA 91125, USA}
\affiliation{AWS Center for Quantum Computing, Pasadena CA 91125}

\begin{abstract}
Quantum many-body systems provide a unique platform for exploring the rich interplay between chaos, randomness, and complexity. In a recently proposed paradigm known as deep thermalization, random quantum states of system $A$ are generated by performing projective measurements on system $B$ following chaotic Hamiltonian evolution acting jointly on $AB$. In this scheme, the randomness of the projected state ensemble arises from the intrinsic randomness of the outcomes when $B$ is measured. Here we propose a modified scheme, in which classical randomness injected during the protocol is converted by quantum chaos into quantum randomness of the resulting state ensemble. We show that for generic chaotic systems this conversion is optimal in that each bit of injected classical entropy generates as much additional quantum randomness as adding an extra qubit to $B$. This significantly enhances the randomness of the projected ensemble without imposing additional demands on the quantum hardware. Our scheme can be easily implemented on typical analog quantum simulators, providing a more scalable route for generating quantum randomness valuable for many applications. In particular, we demonstrate that the accuracy of a shadow tomography protocol can be substantially improved.
\end{abstract}

\maketitle
\begin{figure*}
\centering
\subfloat{%
  \includegraphics[width=0.8\linewidth]{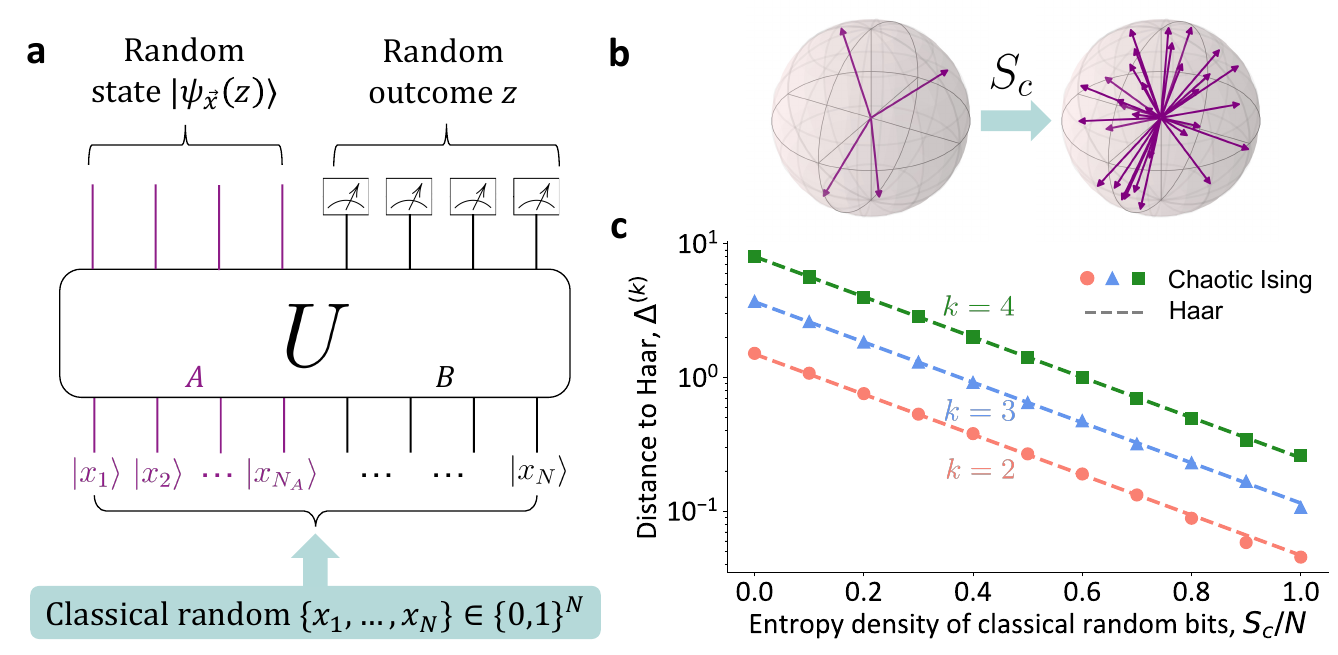}%
}
\caption{Classically-enhanced projected ensembles via random basis state initialization. (a) The system is initialized in a random computational basis state $\bigotimes_{i=1}^{N} \ket{x_i}$, where the $N$-bit string $x = (x_1,\ldots,x_N) \in \{0,1\}^N$ is drawn randomly from a distribution $q(x)$ with entropy $S_{\text{c}}$. The initial state evolves under a fixed unitary $U$, and $N_B$ bath qubits are measured in the computational basis to yield a random outcome $z$. For generic chaotic dynamics, the projected states $\ket{\psi_x(z)}$ on the remaining $N_A$ qubits form a 
projected ensemble $\mathcal{E}$ that approximates a $k$-design. (b) Classical randomness of entropy $S_{\text{c}}$ increases the size of the projected ensemble $\mathcal{E}$ by a factor of up to $ 2^{S_{\text{c}}}$, illustrated on the Bloch sphere. (c) Normalized Hilbert-Schmidt distance $\Delta^{(k)}$ between the $k$-th moment of the projected ensemble $\mathcal{E}$ and the Haar ensemble, against the entropy density $S_{\text{c}}/N$ of the classical distribution $q(x)$, with $N_A = 4$ and $N_B = 6$. The points are obtained numerically by evolving the initial state with the chaotic mixed-field 1D Ising Hamiltonian $H_0$~\eqref{eq:H0_mfim} for a time $JT = 10^3$. The dashed lines denote the analytical root-mean-square distance $\Delta_{\text{rms}}^{(k)}$~\eqref{eq:rms_uniformbasis} when $U$ is a Haar random unitary.}
 \label{fig:reinit_scheme}
\end{figure*}
\textit{Introduction.---} 
Preparing ensembles of random quantum states is an increasingly important task in quantum information science. Quantum randomness is theoretically interesting for its connections to quantum chaos~\cite{cotler2017chaos,hosur2016chaos,roberts2017chaos} and quantum descriptions of black hole dynamics~\cite{page1993information,page2013time,hayden2007black}, and is also practically relevant for a wide range of applications such as randomized benchmarking~\cite{knill2008randomized,dankert2009exact}, quantum communication~\cite{hayden2008decoupling}, phase retrieval~\cite{shelby2017phase}, shadow tomography~\cite{huang2020predicting,mcginley2023shadow}, cryptography~\cite{divincenzo2002quantum,ananth2024simultaneous}, and large-scale device benchmarking~\cite{arute2019quantum,morvan2023phase,zhu2022quantum,shaw2024benchmarking}. However, in general it is difficult to generate truly random quantum states -- known as Haar random states –- which have exponential complexity. Further, for these applications it is typically desirable to create random states on many qubits, but this task is stymied by limitations on the fidelity of near-term quantum processors.

In lieu of creating such truly random states, many applications can already be accomplished through only low-order approximations to the Haar random ensemble, known as $k$-designs~\cite{gross2007evenly,dankert2009exact}, which are statistically indistinguishable from Haar random states for any observable involving up to $k$ copies of the state. Thus, it is useful to characterize the randomness of a quantum state ensemble by quantifying its distance from a $k$-design for $k \geq 2$. 
It is known that approximate $k$-designs can be generated efficiently by random unitary circuits (RUCs)~\cite{brandao2016local,haferkamp2022random,harrow2023approximate,schuster2024random}. Effectively, such RUCs  convert classical randomness (the arrangement of the circuit) into quantum randomness (the generated $k$-design), albeit while requiring a high degree of spatio-temporal control and low noise.

Recently, a new paradigm has emerged for realizing random states without requiring the high-level of control of RUCs, but instead relying on projective measurements following a fixed chaotic Hamiltonian evolution~\cite{Cotler2023emergent,choi2023preparing}, in a process known as deep thermalization~\cite{Cotler2023emergent,choi2023preparing,Ippoliti2022solvable,ho2022exact,wilming2022high,claeys2022emergent,shrotriya2023nonlocality,liu2024deep,bhore2023deep,ippoliti2023dynamical,lucas2023generalized,mark2024maximum,chan2024projected,chang2024deep,varikuti2024unraveling}. In this approach, one evolves a many-qubit quantum system under a chaotic Hamiltonian (with spectral statistics described by random matrix theory), and then measures a subset of the qubits (the bath). For a wide range of physical systems, the resulting state ensemble of the unmeasured qubits (the projected ensemble) approximates a $k$-design~\cite{Cotler2023emergent,choi2023preparing}. This method for generating random states can be used for practical applications such as benchmarking analog quantum simulators~\cite{choi2023preparing}. Since the dynamics is fixed, unlike RUCs, the randomness stems from the intrinsic uncertainty of quantum measurements. While this approach yields exact Haar random states in the limit of infinitely many measured qubits, the convergence to successively higher order $k$-designs is limited by the size of the projectively measured bath, which poses a practical bottleneck for near-term quantum hardware with limited qubit number. 

Here we propose and numerically implement new strategies for circumventing this constraint via injection of classical randomness into quantum chaotic dynamics. This approach allows us to effectively generate higher order $k$-designs for a fixed system size by supplying classical entropy rather than additional bath qubits. We inject classical entropy via either randomizing the initial basis states prepared for the Hamiltonian evolution, or injecting time-independent disorder into the evolution itself. This approach allows us to sample random quantum states with significantly fewer quantum resources than would be required using RUCs or the previously proposed version of deep thermalization. 
In fact, we find that each bit of classical entropy can produce as much additional quantum randomness as adding an extra qubit to the bath, signifying an optimal conversion of classical to quantum randomness. The injected classical randomness can effectively double the size of the bath; hence using the same number of physical qubits we can generate state ensembles that are far closer to the Haar ensemble compared to the projected ensemble protocol proposed in \cite{Cotler2023emergent}. 

Our procedures are experimentally simple to execute on modern quantum platforms, including most analog quantum simulators, and can thus enable the creation of more complex random states on near-term quantum devices. Our approach can improve many applications of quantum randomness~\cite{knill2008randomized,dankert2009exact,hayden2008decoupling,shelby2017phase,huang2020predicting,mcginley2023shadow,divincenzo2002quantum,ananth2024simultaneous,arute2019quantum,morvan2023phase,zhu2022quantum,shaw2024benchmarking}; we describe one in particular, by demonstrating how injecting classical randomness can enhance the performance of shadow tomography~\cite{huang2020predicting,mcginley2023shadow}.

\textit{Projected ensembles and deep thermalization.---}
In the basic framework of the projected ensemble~\cite{Cotler2023emergent,choi2023preparing,goldstein2006distribution,goldstein2016universal}, a fixed initial state of $N$ qubits, taken to be $ |0\rangle^{\otimes N}$, is evolved under a fixed unitary evolution $U_{AB}$. The system is bipartitioned into two subsystems, delineated as $A$ and $B$, of size $N_A$ and $N_B$ qubits, respectively. Qubits in $B$ are then measured in the computational basis producing a bitstring $z \in \{0,1\}^{N_B}$ with probability $p(z)$, which induces a pure quantum state in $A$ conditioned on the measurement outcome. The set of these post-measurement states, together with the probabilities $p(z)$, collectively encode the entire state 
and is known as the projected ensemble.

Our scheme exploits this basic framework, but improves the convergence of the projected ensemble to $k$-designs via the injection of classical randomness. In the simplest case (see Fig.~\ref{fig:reinit_scheme}(a)), this is accomplished by preparing not a fixed initial state for the evolution, but a randomized one. Concretely, the initial state is randomly sampled from the ensemble $\mathcal{E}_{\text{init}} = \{\ket{x},q(x)\}_{x \in \{0,1\}^{N}}$ of computational basis states with probability $q(x)$, and then evolved under the same fixed $U_{AB}$. The bath is then projectively measured, now yielding bitstrings $z$ labeled by the choice of initial state with probabilities $p_x(z) = q(x) \norm{(I_A \otimes \bra{z}_B)U_{AB}\ket{x}}^2$. Accordingly, the post-measurement states in $A$ are labeled by both the bitstring measured in $B$ as well as the initial state, $\ket{\psi_x(z)} \propto (I_A \otimes \bra{z}_B) U_{AB}\ket{x}$. The set of these probabilities and (normalized) post-measurement states then forms the \textit{classically-enhanced} projected ensemble
$\mathcal{E} = \{q(x) p_x(z);\ket{\psi_x(z)}\}$. 

In order to assess the degree of randomness of the projected ensemble $\mathcal{E}$, we compare its moments against that of the Haar ensemble (see Supplemental Materials (SM)~\cite{supp} for more details). The $k$-th moment operator of the ensemble $\mathcal{E}$ is given by
\begin{equation}
    \rho^{(k)} = \sum_{\substack{{x \in \{0,1\}^{N}} \\ {z \in \{0,1\}^{N_B}}}} q(x) p_x(z) (\ket{\psi_x(z)}\bra{\psi_x(z)})^{\otimes k}.
\label{eq:enhanced_projens}
\end{equation}
We then compute the normalized Hilbert-Schmidt distance to the $k$-th moment of the Haar ensemble \change{on $A$},
\begin{equation}
    \Delta^{(k)} = \frac{\norm{\rho^{(k)} - \rho_{\text{Haar}}^{(k)}}_2}{\norm{\rho_{\text{Haar}}^{(k)}}_2}\,,
\label{eq:deltak}
\end{equation}
\change{where $\norm{\cdot}_2$ is the Hilbert-Schmidt norm.} This has an intuitive entropic interpretation since $\Delta^{(k)}$ decreases when the quantum R\'{e}nyi 2-entropy of $\rho^{(k)}$ increases. $\mathcal{E}$ forms an $\epsilon$-approximate state $k$-design if $\Delta^{(k)} \leq \epsilon$, which is consistent with the usual definition using the trace distance~\cite{supp}. The classical randomness of the \change{initial states $\mathcal{E}_{\text{init}}$ on the full system $AB$ can be quantified by $S_c$,}
the R\'{e}nyi 2-entropy of $q(x)$. 
The original protocol in~\cite{Cotler2023emergent,choi2023preparing} corresponds to $S_c = 0$.  


\begin{figure*}
\centering
\subfloat{%
  \includegraphics[width=0.9\linewidth]{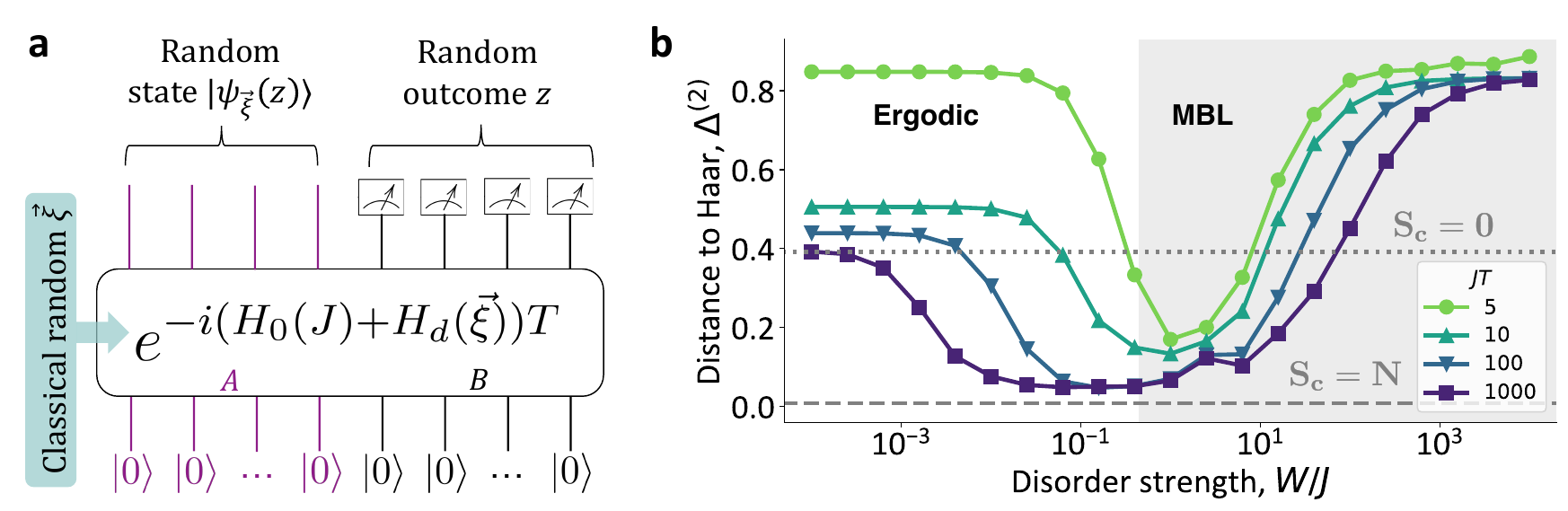}%
}
\caption{(a) Classically-enhanced projected ensembles via random Hamiltonian disorder. The initial state is fixed as $\ket{0}^{\otimes N}$. A spatially inhomogeneous random field $H_d(\vec{\xi}\,)$ with fluctuations of size $W$ is applied on $N$ qubits on top of a fixed chaotic Hamiltonian $H_0$ with interaction energy $J$, and the state is evolved for a quench time $T$. (b) Normalized Hilbert-Schmidt distance $\Delta^{(2)}$ between the projected ensemble $\mathcal{E}$ and the Haar ensemble, against the disorder strength $W/J$. The points are obtained numerically by evolving the initial state under the chaotic (mixed-field) Ising model with disorder~\eqref{eq:mfim_disorder}, for $N_A = 3$ and $ 
 N_B = 8$. A total of $2^{N_A + N_B}$ disorder realizations are sampled, giving a projected ensemble of size $|\mathcal{E}| = 2^{N_A + 2 N_B}$. Numerical results are compared against analytical benchmarks~\eqref{eq:rms_uniformbasis} of $\Delta_{\text{rms}}^{(k)}$ for Haar random unitary evolution with $S_{\text{c}} = 0$ and $N$, denoted by the dotted and dashed lines respectively. The grey shaded region marks the MBL regime, with the ergodic-MBL crossover point at $W/J \approx 0.42$~\cite{supp}. For very large disorder $W/J \gg 1$, the projected ensemble converges to a $1$-design, and $\Delta^{(2)}$ saturates and becomes independent of $N_B$.}
 \label{fig:disorder}
\end{figure*}

\textit{Optimal conversion from classical to quantum randomness.---}
While it is obvious that injecting classical randomness via $\mathcal{E}_{\text{init}}$ can only increase the quantum randomness of the projected ensemble (Fig.~\ref{fig:reinit_scheme}(b)), the interesting question is to what extent quantum randomness can be increased. For analytical tractability, we study the case where $U_{AB}$ is a fixed unitary drawn from the Haar measure on the unitary group $\mathcal{U}(2^{N_A+N_B})$. The quantum randomness in the `typical' case is measured by the root-mean-square distance $\Delta_{\text{rms}}^{(k)} = \sqrt{\E_{U \sim \text{Haar}}\left(\Delta^{(k)}\right)^2}$, obtained by averaging $\left(\Delta^{(k)}\right)^2$ (which is a function of $U_{AB}$) over the Haar measure. Note that the trace distance between $\rho^{(k)}$ and $\rho_{\text{Haar}}^{(k)}$ can be upper bounded by $\bigO{\Delta_{\text{rms}}^{(k)}}$, which we find is exponentially smaller in both $N_A$ and $k$ compared to the bound of Ref.~\cite{Cotler2023emergent} (see Theorem 2 in SM~\cite{supp}). Our main analytical result is stated as follows:
\begin{theorem}[Classically-enhanced projected ensembles]
\label{thm:rms_dist}
Let $\mathcal{E}_{\normalfont\text{init}} = \{q(x),\ket{x}\}$ be an set of orthonormal initial states $\ket{x}$ with probability distribution $q(x)$ and $S_c = -\log_2 \left(\sum_x q(x)^2\right)$ is the R\'{e}nyi 2-entropy of $q(x)$. For any \change{$2 \leq k < 2^{(N_A+N_B)/4}$}, the root-mean-square distance $\normalfont\Delta_{\text{rms}}^{(k)}$ between the $k$-th moments of the projected ensemble $\mathcal{E}$ and the Haar ensemble approaches
\begin{equation}\normalfont
    \left(\Delta_{\text{rms}}^{(k)}\right)^2 = \frac{1}{k! 2^{S_{\text{c}}+N_B- kN_A}}
\label{eq:rms_uniformbasis}
\end{equation}
 as $N_A$ and $N_B \to \infty$.
\end{theorem}
The detailed proof of Theorem~\ref{thm:rms_dist} using Weingarten calculus is provided in the SM~\cite{supp}. From Eq.~\eqref{eq:rms_uniformbasis}, we can see that quantum randomness is maximized when $\mathcal{E}_{\text{init}}$ forms a $1$-design, i.e., $S_c = N_A + N_B$ corresponding to the maximum amount of classical entropy in the initial state distribution. \change{The expression for $k=1$ is trivial and derived in the SM~\cite{supp}.} This reveals explicitly the conversion of classical to quantum randomness. Moreover, since $\Delta_\text{rms}^{(k)}$ vanishes exponentially with $N_B + S_c$, Theorem~\ref{thm:rms_dist} implies an (asymptotically) optimal conversion from classical to quantum randomness, where injecting $S_{\text{c}}$ bits of classical entropy yields the same $\Delta_{\text{rms}}^{(k)}$ as adding $S_{\text{c}}$ bath qubits without classical randomness. In other words, classical randomness and bath qubits are interconvertible resources from the perspective of projected ensembles. 

This is particularly advantageous for experimental implementation since $N_B$ is often limited on the quantum hardware due to noise and practical constraints. Additionally, up to $N_B$ bits of classical entropy can be obtained naturally by reusing the bath measurement outcome $z$ as an initial state of $B$ in a subsequent run of the protocol. 
Equation~\eqref{eq:rms_uniformbasis} indicates that injecting classical randomness can reduce the distance between the projected ensemble and the Haar ensemble by a factor exponential in $N=N_A+N_B$ at a very modest cost for the quantum hardware. 
Thus injecting classical randomness significantly reduces the quantum resources needed to sample random quantum states on analog simulators.

Since $S_{\text{c}}$ is at most $N_A+N_B$ bits, this gives an effective bath size of up to $N_A + 2N_B$ qubits, essentially increasing the bath size by more than a factor of 2. This allows one to generate higher-order designs with a maximal achievable $k$ almost twice as large~\cite{supp} as in the original protocol~\cite{Cotler2023emergent,choi2023preparing}. For example, $2$-designs can be generated on $A$ with $N_B < N_A$. In contrast, not even a $1$-design can be achieved for $N_B < N_A$ if no classical randomness is injected.   


We now show numerically that Theorem~\ref{thm:rms_dist} holds even for Hamiltonian dynamics. We plot $\Delta^{(k)}$ against the classical entropy density $S_{\text{c}}/N$ in Fig.~\ref{fig:reinit_scheme}(c) for $U_{AB} = \exp(-iH_0 T)$ with quench time $JT=10^3$, where $H_0$ is the mixed-field 1\text{D} Ising Hamiltonian
\begin{equation}
    H_0 = \sum_{i=1}^{N_A+N_B} (h_x X_i + h_y Y_i + J X_i X_{i+1})
\label{eq:H0_mfim}
\end{equation}
with open boundary conditions. $X_i, Y_i$ and $Z_i$ are Pauli operators for qubit $i$. The computational basis states are eigenstates of $Z_i$, and thus satisfy $\braket{H_0} = 0$ which is necessary for deep thermalization at infinite temperature. We choose $\{h_x, h_y, J\} = \{0.8090,0.9045,1\}$ in the non-integrable and chaotic regime~\cite{kim2013ballistic,kim2014testing}, which was also numerically demonstrated to produce random projected ensembles~\cite{Cotler2023emergent} for $S_{\text{c}} = 0$. We observe that $\Delta^{(k)}$ decreases exponentially with the classical entropy $S_{\text{c}}$, in excellent agreement with our analytical formula~\eqref{eq:rms_uniformbasis} derived for Haar random dynamics. This suggests that the near-optimal conversion from classical to quantum randomness is a more generic feature of quantum chaos.

\textit{Classical to quantum randomness using Hamiltonian disorder.---}
Another way to inject classical randomness is by fixing the initial state and applying a disorder Hamiltonian $H_d(\vec{\xi}\,)$ with $2^{S_c}$ independent disorder realizations $\vec{\xi}$, on top of a fixed chaotic Hamiltonian $H_0$~\cite{perrin2024dynamic}; see Fig.~\ref{fig:disorder}(a). Such a setup is motivated both by experimental considerations of analog quantum simulators which may more readily apply time-constant disordered potentials rather than localized bit-flip operations, and also for its theoretical connections to many-body localization, as we shall explore below. 

As an example, we consider the Hamiltonian 
\begin{equation}
H = H_0 + H_d(\vec{\xi}\,)=H_0+\sum_{i=1}^{N_A+N_B}\xi_i X_i
\label{eq:mfim_disorder}
\end{equation}
where $H_0$ is given by Eq.~\eqref{eq:H0_mfim} and disorder $\xi_i \overset{\text{i.i.d.}}{\sim} \text{Uniform}[-W,W]$ with disorder strength $W/J$. Figure~\ref{fig:disorder}(b) shows the behavior of $\Delta^{(2)}$ against disorder strength $W/J$ (similar behavior is observed for other $k > 2$), with $2^{N_A + N_B}$ disorder realizations. The analytical values of $\Delta_{\text{rms}}^{(k)}$ in Eq.~\eqref{eq:rms_uniformbasis} for $S_{\text{c}} = 0$ and $S_c = N$ are indicated by dashed lines as a benchmark. At very weak disorder strengths $W/J \to 0$, $\Delta^{(k)}$ converges for large $JT$ near the benchmark value with $S_{\text{c}} = 0$, consistent with previous results~\cite{Cotler2023emergent}. As the disorder strength increases, classical randomness gets converted into quantum randomness, causing $\Delta^{(k)}$ to decrease. At sufficiently long evolution times, $\Delta^{(k)}$ can become close to the benchmark value with $S_{\text{c}} = N$. This signifies a near-optimal conversion from classical to quantum randomness. Our analytical formula~\eqref{eq:rms_uniformbasis} works well here even though classical randomness is injected via the dynamics instead of the initial state, which demonstrates the generality of our protocol. 

At strong disorder strengths $W/J \gg 1$, $\Delta^{(k)}$ increases and saturates. We attribute this behavior to the fact that for strong disorder, the projected ensemble behaves like an ensemble of random product states (at best a low-randomness $1$-design), due to many-body localization effects which become relevant when $W/J \gtrsim 1$. We note that a more experimentally accessible scheme of adding a random global detuning $H_d = \xi \sum_i X_i$ where $\xi \sim \text{Uniform}[-W,W]$ (instead of spatially inhomogeneous disorder) also yields qualitatively similar results (not shown). 

The crossover between the benchmark values at $S_{\text{c}} = 0$ and $S_{\text{c}} = N$ as disorder strength increases can be roughly estimated via a simple argument. Firstly, we need a quench time $JT \gtrsim N_A$ to get volume-law entanglement between the system and bath qubits. To contribute appreciably to the randomness of the projected ensemble, the effects of disorder must be integrated over a time $T$ such that $WT\gtrsim 1$. On the other hand, for efficient conversion of classical to quantum randomness, we must avoid the many-body localized regime $W/J\gtrsim 1$~\cite{supp}. Therefore, for $JT\gg N_A$ we expect nearly maximal conversion of classical to quantum randomness for $1/JT \lesssim W/J\lesssim 1$. The behavior of $\Delta^{(2)}$ shown in Fig.~\ref{fig:disorder}(b) is consistent with this expectation. An interesting future direction would be to characterize the many-body localization transition via the projected ensemble.


\textit{Application: Classical shadow tomography.---}
A practical application of our protocol is classical shadow tomography~\cite{huang2020predicting} for learning expectation values of observables in unknown states. The state is scrambled with a unitary, followed by measurements in the computational basis. 
From the outcomes and the inverted scrambling dynamics, one can construct a classical representation of the unknown state, which can be used to accurately estimate the expectation values of many observables. In Refs.~\cite{tran2023measuring,mcginley2023shadow}, it was proposed to use projected ensembles to generate the scrambling dynamics for shadow tomography.
However, the estimation accuracy  
crucially depends on the quantum randomness of the projected ensemble~\cite{mcginley2023shadow}.
Now we show that one can gain an exponential increase in accuracy by adding classical randomness, without incurring extra cost on the quantum computer. 

The initial state is $\rho_A \otimes \ket{x}\bra{x}_B$, where $\rho_A$ is the unknown state to be learned. Classical randomness is injected by randomly initializing the $B$ subsystem in computational basis states $\ket{x}$, similar to the setup in Fig.~\ref{fig:reinit_scheme}(a), up to a maximum of $S_{\text{c}} = 2^{N_B}$ bits. A unitary $U$ is then applied on the full system. The projected ensemble $\mathcal{E}$ is constructed by measuring subsystem $B$ in the computational basis $\{z_B\}$. The states in $\mathcal{E}$ are then measured in the computational basis $\{z_A\}$. The measurement outcomes $(z_A,z_B)$ occur with probability $p_x(z_A,z_B) =\bra{z_A,z_B} U (\rho_A \otimes \ket{x}\bra{x})U^\dag \ket{z_A,z_B}$. In the classical post-processing, we construct the shadow estimator as
\begin{equation}
    \hat{\rho}_{A,x} = \frac{\left(2^{N_A} + 1\right)\bra{x}U^\dag  \ket{z_A,z_B}\bra{z_A,z_B}U \ket{x}}{\text{Tr}_A(\bra{x}U^\dag  \ket{z_A,z_B}\bra{z_A,z_B}U \ket{x})} - I_A
\end{equation}
which satisfies the normalization $\text{Tr} (\hat{\rho}_{A,x}) = 1$ \change{(see SM~\cite{supp} for a detailed explanation of the protocol). Intuitively, $\hat{\rho}_{A,x}$ is constructed to `undo' the scrambling of $\rho_A$, in order to estimate $\rho_A$}. This gives the estimator \change{$\hat{O} = \Tr (O\hat{\rho}_{A,x}) \approx \Tr(O\rho_A)$} for the observable $O$, averaged over $L$ measurement shots. The shadow estimator can be analogously defined if classical randomness is instead injected by adding random disorder such as in Eq.~\eqref{eq:mfim_disorder} to the dynamics for a fixed initial state. The estimation error is $\delta O = |\hat{O} - \Tr(O\rho_A)|$, with the bias (i.e., systematic) error given by $\delta O$ as $L \to \infty$.
$\hat{\rho}_{A}$ is an unbiased estimator of $\rho_A$ if the projected ensemble forms an exact $2$-design~\cite{mcginley2023shadow}. For approximate $2$-designs, the estimation incurs a bias error that grows with the distance $\Delta^{(2)}$ from a $2$-design. 
Thus, from Theorem~\ref{thm:rms_dist}, we expect the bias error to be exponentially reduced by increasing $S_c$. 

\begin{figure}
\centering
\includegraphics[width=0.47\textwidth]{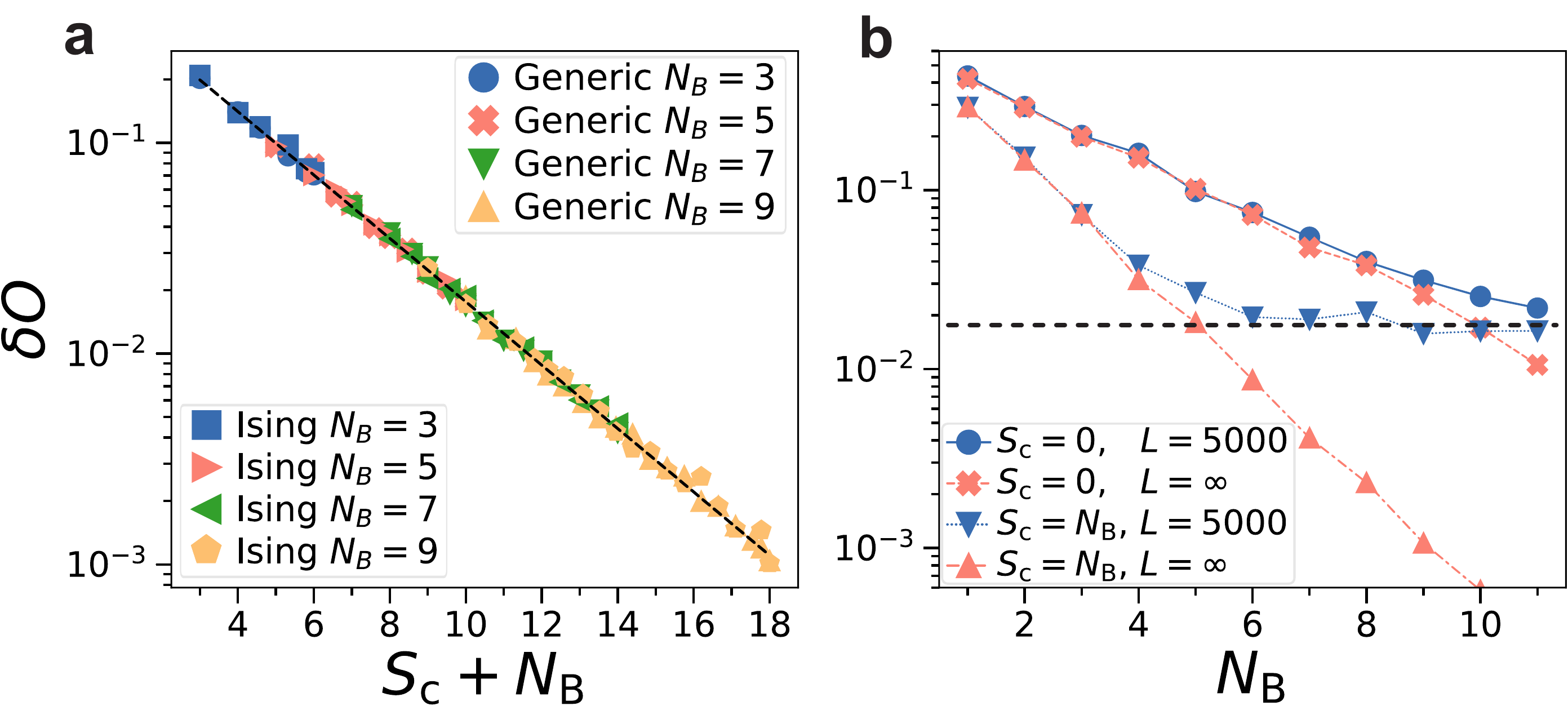}
\caption{(a) Average bias error $\delta O$ (with $L \to \infty$ measurement shots) against classical entropy $S_{\text{c}}$ plus number of bath qubits $N_\text{B}$. For `Generic', we randomize the initial state of the $N_\text{B}$ bath qubits and evolve with a fixed Haar random unitary. 
For `Ising', we evolve with~\eqref{eq:mfim_disorder} for $JT=100$,  where the initial state is fixed and the disorder of the Hamiltonian is randomized for $2^{S_c}$ realizations with $W/J=0.5$. Dashed line depicts $\delta O \propto 2^{-(S_c+N_B)/2}$.
Error is nearly independent of whether randomness originates from measuring bath qubits, adding random initialization, or disorder in the time-evolution.
(b) Error $\delta O$ against $N_\text{B}$ for different $L$ and $S_{\text{c}}$. Black dashed line shows the error for $L=5000$ of an unbiased classical shadow protocol~\cite{huang2020predicting}.
In both plots, we have $N_\text{A}=1$ and average over 100 randomly chosen instances of unitaries to suppress statistical fluctuations (see SM~\cite{supp} for details).}
 \label{fig:shadow}
\end{figure}

This is demonstrated in Fig~\ref{fig:shadow}(a), which shows the average bias error against the classical entropy $S_{\text{c}}$ injected into the projected ensemble, with randomly chosen Pauli operators $O$. 
We numerically simulate two different settings of classical randomness: (I) Bath qubits randomly initialized from the uniform distribution over $S_{\text{c}}$ classical bits, with $U$ being a fixed Haar random unitary. (II) Bath qubits fixed in the $\ket{0}^{\otimes N_B}$ state and evolved under the disordered Hamiltonian~\eqref{eq:mfim_disorder} sampled from a set of $2^{S_{\text{c}}}$ randomly chosen disorder realizations~\cite{mcginley2023shadow}.

Similar to our analytical prediction for $\Delta_{\text{rms}}^{(k)}$ in Theorem~\ref{thm:rms_dist}, we observe that the bias error decreases exponentially as $\sim 2^{-(N_B + S_\text{c})/2}$ in both settings of classical randomness. Injecting $S_{\text{c}}$ bits of entropy yields nearly the same reduction in bias error as having additional $S_{\text{c}}$ bath qubits without classical randomness, arising from the near-optimal conversion from classical to quantum randomness. 
In Fig.~\ref{fig:shadow}(b), we plot $\delta O$ against $N_B$, with classical randomness in the setting (I). We find that for a finite $L$, the error converges at large $N_\text{B}$ to that expected of an unbiased shadow estimator (black dashed line)~\cite{huang2020predicting}. Increasing $S_c$ causes the error to decrease significantly. Our protocol remains computationally efficient and scalable at large $N_A$, using the `patched quench' setup~\cite{tran2023measuring}, see SM~\cite{supp} for details. 

\textit{Discussion.---} 
We have shown how quantum chaos can convert classical randomness into quantum randomness, by injecting classical entropy into the deep thermalization framework. For generic chaotic dynamics, each bit of classical entropy generates nearly as much quantum randomness as an additional bath qubit. From a practical viewpoint, injecting classical entropy allows one to improve the generation of approximate $k$-designs, with the maximum achievable $k$ almost doubled, at only a small additional cost to the quantum hardware. This enhancement in quantum randomness directly translates to better performance for applications which utilize random quantum dynamics as a resource, as we have demonstrated through the example of shadow tomography. Our scheme is easier to implement experimentally on many analog quantum simulators, as compared to well-established schemes using RUCs, by leveraging the inherent complexity of many-body quantum dynamics. From a theoretical perspective, our work raises interesting questions about whether injecting classical randomness improves the convergence of other related protocols such as finite-temperature projected ensembles and temporal ensembles to their respective maximum entropy ensembles~\cite{mark2024maximum}, about the cross-over between ergodic and localized behavior, and about the universal features of quantum chaos.

\begin{acknowledgements}
\textit{Acknowledgements.---}
We thank Abhishek Anand, Jielun Chen, Wen Wei Ho, Daniel Mark, and Richard Tsai for helpful discussions. We acknowledge support from the DARPA ONISQ program (W911NF2010021), the DOE (DE-SC0021951), the Army Research Office MURI program (W911NF2010136), the NSF CAREER award (1753386), the Institute for Quantum Information and Matter, an NSF Physics Frontiers Center (NSF Grant PHY-1733907), and the Technology Innovation Institute (TII).
J.P. also acknowledges support from the U.S. Department of Energy Office of Science, Office of Advanced Scientific Computing Research (DE-NA0003525, DE-SC0020290), and the U.S. Department of Energy, Office of Science, National Quantum Information Science Research Centers, Quantum Systems Accelerator.
%

\end{acknowledgements}
 \let\oldaddcontentsline\addcontentsline
\renewcommand{\addcontentsline}[3]{}

\bibliography{bib}

\let\addcontentsline\oldaddcontentsline

\appendix

\onecolumngrid
\newpage

\setcounter{secnumdepth}{2}
\setcounter{figure}{0}
\renewcommand{\thefigure}{S\arabic{figure}}


\end{document}


\author{Wai-Keong Mok}
\affiliation{Institute for Quantum Information and Matter, California Institute of Technology, Pasadena, CA 91125, USA}
\author{Tobias Haug}
\affiliation{Quantum Research Centre, Technology Innovation Institute, Abu Dhabi, UAE}
\author{Adam L. Shaw}
\thanks{Present address: Department of Physics, Stanford University, Stanford, CA}
\affiliation{Institute for Quantum Information and Matter, California Institute of Technology, Pasadena, CA 91125, USA}

\author{Manuel Endres}
\affiliation{Institute for Quantum Information and Matter, California Institute of Technology, Pasadena, CA 91125, USA}
\author{John Preskill}
\affiliation{Institute for Quantum Information and Matter, California Institute of Technology, Pasadena, CA 91125, USA}
\affiliation{AWS Center for Quantum Computing, Pasadena CA 91125}

\title{Supplemental Material for ``Optimal Conversion from Classical to Quantum Randomness via Quantum Chaos"}

\maketitle

\setcounter{theorem}{1}
\makeatletter
\renewcommand{\thefigure}{S\arabic{figure}}
We provide additional technical details and data supporting the claims in the main text. 

\tableofcontents

\section{Quantifying quantum randomness}
\label{app:quantifying_qrandom}
Here, we briefly introduce the various measures of quantum randomness and the relationship between them. Given an ensemble of states $\mathcal{E} = \{p_i,\ket{\psi_i}\}_i$ with probabilities $p_i$, we can construct the density operator
\begin{equation}
    \rho^{(k)} \equiv \Eset{\ket{\psi}_i \sim \mathcal{E}}\left[(\ket{\psi_i}\bra{\psi_i})^{\otimes k}\right] 
= \sum_i p_i (\ket{\psi_i}\bra{\psi_i})^{\otimes k}
\end{equation} 
for the $k$-th moment of the ensemble $\mathcal{E}$. Analogously, the $k$-th moment of the Haar ensemble $\mathcal{E}_{\text{Haar}}$ is denoted as 
\begin{equation}
    \rho^{(k)}_{\text{Haar}} = \int_{\text{Haar}(d)} d\psi (\ket{\psi}\bra{\psi})^{\otimes k}.
\end{equation}
The Haar measure on the space of pure states (also known as the Fubini-Study metric) is induced by the Haar measure on the unitary group $\mathcal{U}(d)$, where $d$ is the dimension of the Hilbert space. For any ensemble $\mathcal{E}$, $\rho^{(k)}$ lies in the symmetric subspace of the $k$-fold Hilbert space $\mathcal{H}^{\otimes k}$, with dimension $D_k$ given by~\cite{harrow2013church}
\begin{equation}
    D_k = {{d+k-1}\choose{k}}.
\label{eq:Dk}
\end{equation} 
In particular,
\begin{equation}
    \rho_{\text{Haar}}^{(k)} = \frac{1}{D_k} \hat{P}_\text{sym} = \frac{(d-1)!}{(d+k-1)!} \sum_{\pi \in S_k} \hat{\pi} 
\label{eq:rho_haar_defn}
\end{equation}
is proportional to the orthogonal projector $\hat{P}_{\text{sym}}$ onto the symmetric subspace, which can be expressed as a uniform sum over the (unitary) permutation operators $\hat{\pi}$ which permute the tensor factors between the $k$ copies, corresponding to each element $\pi$ of the symmetric group $S_k$. Equation~\eqref{eq:rho_haar_defn} can be derived using representation theory and the Schur-Weyl duality~\cite{harrow2013church}.

To measure the quantum randomness of the ensemble $\mathcal{E}$ (i.e., closeness to the Haar ensemble), we use the normalized distance between $\rho^{(k)}$ and $\rho^{(k)}_\text{Haar}$, defined as
\begin{equation}
    \Delta_{\alpha}^{(k)} = \frac{\norm{\rho^{(k)} - \rho_{\text{Haar}}^{(k)}}_\alpha}{\norm{\rho_{\text{Haar}}^{(k)}}_\alpha},
\label{eq:delta_alphas}
\end{equation}
with $\norm{A}_\alpha = (\text{Tr} |A|^{\alpha})^{1/\alpha}$ the Schatten $\alpha$-norm, for $\alpha \geq 1$. This gives a family of distance measures parameterized by the Schatten index $\alpha$. The special cases of $\alpha = 1, 2, \infty$ correspond to the trace, Hilbert-Schmidt and operator norms respectively. In~\cite{ippoliti2023dynamical}, it was shown that $\Delta_\alpha^{(k)}$ satisfies monotonicity, i.e., $\Delta_{\alpha}^{(k)} \leq \Delta_\alpha^{(k+1)}$. This means that if we define an ensemble $\mathcal{E}$ as an $\epsilon$-approximate $k$-design when $\Delta_\alpha^{(k)} \leq \epsilon$, it is automatically an $\epsilon$-approximate $k^\prime$-design for $k^\prime < k$ with the same Schatten index $\alpha$. Thus, $\Delta_\alpha^{(k)}$ provides a sensible definition of approximate state designs for any $\alpha \geq 1$: 
\begin{definition}[$\epsilon$-approximate state $k$-design]
\label{defn:approxdesign}
   A state ensemble $\mathcal{E} = \{p_i, \ket{\psi_i}\}$ is an $\epsilon$-approximate state $k$-design under the Schatten $\alpha$-norm, for $\alpha \geq 1$, if
   \begin{equation}\normalfont
       \Delta_\alpha^{(k)} \equiv     \frac{\norm{\rho_{\mathcal{E}}^{(k)} - \rho_{\text{Haar}}^{(k)}}_\alpha}{\norm{\rho_{\text{Haar}}^{(k)}}_\alpha} \leq \epsilon.
   \end{equation}
\end{definition}
$\rho^{(k)}$ is said to be an \textit{exact state $k$-design} if and only if $\Delta_\alpha^{(k)} = 0$. We now derive a relationship between the different distances labelled by $\alpha$, showing that they are equivalent, up to factors of dimension $D_k$ of the symmetric subspace.
\begin{proposition}
\label{prop:dist_inequality}
    The normalized distances $\Delta_\alpha^{(k)}$ obeys the inequality
\begin{equation}
    \Delta_\alpha^{(k)} \leq \Delta_{\alpha+1}^{(k)} \leq D_k^{\frac{1}{\alpha(\alpha+1)}}  \Delta_\alpha^{(k)},
\label{eq:relation_bt_alphanorm}
\end{equation}
where $D_k = {d+k-1 \choose k}$ is the dimension of the symmetric subspace of the $k$-fold Hilbert space $\mathcal{H}^{\otimes k}$.
\end{proposition}
\begin{proof}
    Observe that
\begin{equation}
    \norm{\rho_{\text{Haar}}^{(k)}}_\alpha = (\Tr (\rho_{\text{Haar}}^{(k) \alpha}) )^{1/\alpha} = D_k^{\frac{1-\alpha}{\alpha}},
\end{equation}
using the fact that $\rho_{\text{Haar}}^{(k)}$ is proportional to the orthogonal projector onto the symmetric subspace; see Eq.~\eqref{eq:rho_haar_defn}. Thus,
\begin{equation}
\norm{\rho_{\text{Haar}}^{(k)}}_{\alpha+1} = D_k^{-\frac{\alpha}{\alpha+1}} = \norm{\rho_{\text{Haar}}^{(k)}}_{\alpha} D_k^{-\frac{1}{\alpha(\alpha+1)}}.
\label{eq:relation_bt_alphanorm_haar}
\end{equation}
Now, using the equivalence of Schatten norms, and $\text{rank}(\rho^{(k)} - \rho_{\text{Haar}}^{(k)}) \leq D_k$, we have
\begin{equation}
    \norm{\rho^{(k)} - \rho_{\text{Haar}}^{(k)}}_{\alpha+1} \leq \norm{\rho^{(k)} - \rho_{\text{Haar}}^{(k)}}_{\alpha} \leq D_k^{\frac{1}{\alpha(\alpha+1)}} \norm{\rho^{(k)} - \rho_{\text{Haar}}^{(k)}}_{\alpha+1}.
\end{equation}
which can be inverted to give
\begin{equation}
    D_k^{-\frac{1}{\alpha(\alpha+1)}}\norm{\rho^{(k)} - \rho_{\text{Haar}}^{(k)}}_{\alpha} \leq \norm{\rho^{(k)} - \rho_{\text{Haar}}^{(k)}}_{\alpha+1} \leq \norm{\rho^{(k)} - \rho_{\text{Haar}}^{(k)}}_{\alpha}.
\end{equation}
Dividing through by $\norm{ \rho_{\text{Haar}}^{(k)}}_{\alpha+1}$ and using Eq.~\eqref{eq:relation_bt_alphanorm_haar} gives the desired result Eq.~\eqref{eq:relation_bt_alphanorm}.
\end{proof}
Proposition~\ref{prop:dist_inequality} implies that if the ensemble $\mathcal{E}$ forms an $\epsilon$-approximate $k$-design under the Schatten-$\alpha$ norm, it is also an $\epsilon$-approximate $k$-design under the Schatten-$\alpha^\prime$ norm, for any  $\alpha^\prime \leq \alpha$. In particular, substituting $\alpha = 1$ in Proposition~\ref{prop:dist_inequality}, we have the following bound on the trace distance,
\begin{equation}
    D_k^{-1/2} \Delta_2^{(k)} \leq \Delta_1^{(k)} \leq \Delta_2^{(k)}.
\label{eq:bounds_td}
\end{equation}

\subsection{Operational meaning of distance measures}
For $\alpha = 1$, $\Delta_1^{(k)}$ is exactly twice the trace distance between $\rho^{(k)}$ and $\rho_{\text{Haar}}^{(k)}$. This has a nice operational definition using the Holevo-Helstrom theorem~\cite{holevo1973statistical,helstrom1969quantum}
\begin{equation}
    P_\text{succ,k} = \frac{1}{2} + \frac{\Delta_1^{(k)}}{4},
\end{equation}
where $P_\text{succ,k}$ is the optimal single-shot success probability of distinguishing $\mathcal{E}$ from the Haar ensemble $\mathcal{E}_\text{Haar}$, given $k$ copies of the quantum states. Thus, $\Delta_1^{(k)}$ is directly related to the indistinguishability between $\mathcal{E}$ and the Haar ensemble. For general $\alpha \geq 1$, we can assign an operational meaning to $\Delta_\alpha^{(k)}$ in terms of the deviation of expectation values:
\begin{proposition}
    Let $A \in \mathcal{L}(\mathcal{H}^{\otimes k})$ satisfying $\lVert{A\rVert}_\beta = 1$ be arbitrary, and $\alpha^{-1} + \beta^{-1} = 1$. If $\mathcal{E}$ is an $\epsilon$-approximate state $k$-design under the Schatten $\alpha$-norm, then
    \begin{equation}\normalfont
        \left|\braket{A}_{\mathcal{E}} - \braket{A}_{\mathcal{E}_\text{Haar}}\right| \leq \Delta_\alpha^{(k)}D_k^{\frac{1-\alpha}{\alpha}} \leq \epsilon D_k^{\frac{1-\alpha}{\alpha}},
    \end{equation}
where $\braket{A}_\mathcal{E}$ is the ensemble average of $\braket{A}$ over $\mathcal{E}$.
\end{proposition}
\begin{proof}
    By the duality of the Schatten norm,
    \begin{equation}
        \Delta_\alpha^{(k)} D_k^{\frac{1-\alpha}{\alpha}} = \norm{\rho^{(k)}-\rho_{\text{Haar}}^{(k)}}_\alpha = \sup_{\norm{A}_\beta \leq 1} \left| \text{Tr}\left(A \left(\rho^{(k)}-\rho_\text{Haar}^{(k)}\right)\right)\right| = \sup_{\norm{A}_\beta \leq 1} \left| \braket{A}_{\mathcal{E}} - \braket{A}_{\mathcal{E}_\text{Haar}}\right|.
    \end{equation}
\end{proof}
As an example, if $R$ is a Pauli string acting on $nk$ qubits (here $d = 2^n$), $\lVert R \rVert_\beta = 2^{nk/\beta} = 2^{nk (\alpha-1)/\alpha}$, so
\begin{equation}
    \left| \braket{R}_{\mathcal{E}} - \braket{R}_{\mathcal{E}_\text{Haar}}\right| \leq \epsilon \left(\frac{2^{n k}}{D_k}\right)^{\frac{\alpha-1}{\alpha}} \sim \epsilon (k!)^{\frac{\alpha-1}{\alpha}} \quad (k \ll 2^n).
\end{equation}

\subsection{Frame potential}
\label{app:framepot}
For $\alpha = 2$ in Eq.~\eqref{eq:delta_alphas}, the normalized Hilbert-Schmidt distance $\Delta_2^{(k)}$ is often simpler to compute than $\alpha = 1$. This can be seen by rewriting the $\Delta_2^{(k)}$ in terms of the frame potential $F^{(k)} = \text{Tr}\left({\rho^{(k)}}^2\right)$,
\begin{equation}
    \Delta_2^{(k)} = \left( \frac{F^{(k)}}{F_{\text{Haar}}^{(k)}} - 1 \right)^{1/2}.
\label{eq:delta2}
\end{equation}
From Eq.~\eqref{eq:rho_haar_defn}, the frame potential of the Haar ensemble is given by
\begin{equation}
    F_{\text{Haar}}^{(k)}(d) = {d+k-1\choose{k}}^{-1} = \frac{1}{D_k}.
\end{equation}
For the ensemble $\mathcal{E} = \{p_i,\ket{\psi_i}\}_i$, the frame potential is
\begin{equation}
    F^{(k)} = \E_{i,j \sim \mathcal{E}} \sparens{|\braket{\psi_i|\psi_j}|^{2k}} = \sum_{i,j = 1}^{|\mathcal{E}|} p_i p_j |\braket{\psi_i|\psi_j}|^{2k} 
\end{equation}
which is often more convenient to compute numerically (as well as analytically) without having to construct $\rho^{(k)}$ explicitly. For numerical calculations where the cardinality $|\mathcal{E}|$ of the ensemble is large, $F^{(k)}$ can be estimated using Monte Carlo integration, by sampling from the ensemble.

The frame potential $F^{(k)}$ also has a nice physical interpretation as it can be related to the quantum R\'{e}nyi 2-entropy of $\rho^{(k)}$:
\begin{equation}
    F^{(k)} = 2^{-S_2(\rho^{(k)})}.
\label{eq:framepot_renyi}
\end{equation}
Intuitively, a lower frame potential implies a higher quantum R\'{e}nyi 2-entropy $S_2(\rho^{(k)})$ and thus more quantum randomness. In the main text, we use $\Delta^{(k)} \equiv \Delta_2^{(k)}$ by default and omit the subscript.

\subsection{Root-mean-square distance}

In the main text, we quantify the randomness of the projected ensemble via the root-mean-square distance
\begin{equation}
    \Delta_{\text{rms}}^{(k)} = \sqrt{\E\parens{\Delta_2^{(k)}}^2},
\end{equation}
where the average is taken over the unitary $U$, drawn from some distribution such as the Haar measure on the unitary group. This can be related to conventional measures of quantum randomness based on the trace distance $\Delta_1^{(k)}$ (we omit the factor of $1/2$ for convenience). In particular, we have
\begin{equation}
    \E \Delta_1^{(k)} \leq \E \Delta_2^{(k)} \leq \Delta_{\text{rms}}^{(k)},
\label{eq:TD_rms_bound}
\end{equation}
where the first inequality comes from Eq.~\eqref{eq:bounds_td}. Thus, the root-mean-square distance sets an upper bound for the average trace distance. Moreover, we can use $\Delta_{\text{rms}}^{(k)}$ to obtain a probabilistic bound on $\Delta_1^{(k)}$, for a single instance of $U$.
\begin{proposition}
\label{prop:markovbound_TD}
    Let $U$ be sampled from a distribution $\mathcal{D}$ of unitaries, and $\mathcal{E}$ be any state ensemble. The root-mean-square (normalized) Hilbert-Schmidt distance between the $k$-th moment $\rho^{(k)}$ of the ensemble $\mathcal{E}$ and the $k$-th moment $\rho_{\text{Haar}}^{(k)}$ of the Haar ensemble is denoted as $\Delta_{\text{rms}}^{(k)}$. With probability at least $1 - \delta$, the trace distance $\Delta_1^{(k)}$ is upper bounded by
    \begin{equation}
        \Delta_1^{(k)} \leq \frac{\Delta_{\text{rms}}^{(k)}}{\delta}.
    \end{equation}
Equivalently, $\mathcal{E}$ forms an $(\Delta_{\text{rms}}^{(k)}/\delta)$-approximate state $k$-design, under the Schatten-$1$ norm, by Definition~\ref{defn:approxdesign}.
\end{proposition}
\begin{proof}
By Markov's inequality,
\begin{equation}
    \Pr_{U \sim \mathcal{D}}\sparens{\Delta_1^{(k)} \geq \epsilon} \leq \frac{\Eset{U \sim \mathcal{D}}\sparens{\Delta_1^{(k)}}}{\epsilon} \leq \frac{\Delta_{\text{rms}}^{(k)}}{\epsilon},
\end{equation}
where we used Eq.~\eqref{eq:TD_rms_bound}. Setting the right hand side as $\delta$ gives the desired result.
\end{proof}
Proposition~\ref{prop:markovbound_TD} is valid for any ensemble $\mathcal{E}$. For the particular case of the projected ensemble, we can combine Proposition~\ref{prop:markovbound_TD} and Theorem 1 (in the main text) to obtain an explicit bound for the trace distance.
\begin{theorem}[Projected ensembles form approximate $k$-designs]
    Let $U$ be sampled from the Haar measure on the unitary group $\mathcal{U}(2^{N_A + N_B})$. With probability at least $1 - \delta$, the projected ensemble $\mathcal{E}$ on $N_A$ qubits with classical entropy $S_c$ (as defined in the main text) forms an $\epsilon$-approximate state $k$-design under the Schatten-$1$ norm, where $k < 2^{(N_A + N_B)/4}$, and
\begin{equation}
    \epsilon^2 \leq \frac{1}{\delta^2 k! 2^{S_c + N_B - k N_A}},
\label{eq:rms_tdbound}
\end{equation}
as $N_A, N_B \to \infty$.
\end{theorem}
Comparing Eq.~\eqref{eq:rms_tdbound} (at $S_c = 0$) with the bound from Ref.~\cite{Cotler2023emergent},
\begin{equation}
    \epsilon^2 = \bigO{\frac{2k-1}{2^{N_B - 4 k N_A}} (2 k N_A + \log(1/\delta)},
\label{eq:cotler_TDbound}
\end{equation}
we can see that our bound on $\epsilon$ has an improved scaling that is exponentially smaller in $N_A$ and $k$, if $\delta = \Omega(2^{-k N_A})$. 

For a constant $k$, Eq.~\eqref{eq:cotler_TDbound} from Ref.~\cite{Cotler2023emergent} requires $N_B - 4k N_A \gg 1$ to guarantee closeness to a $k$-design. On the other hand, our result~\eqref{eq:rms_tdbound} suggests that $N_B - k N_A \gg 1$ is sufficient. Our result implies that the number of bath qubits needed for the projected ensemble to form an approximate $k$-design (with high probability) is roughly a factor of $4$ smaller than previously expected.

\section{Large-$d$ asymptotics of Weingarten sums}
In the regime of large Hilbert space dimension $d$ (typically exponential in the number of degrees of freedom), the sums involving Weingarten functions can be well approximated by keeping only the leading order contribution. Intuitively, the Weingarten functions $\text{Wg}(\sigma^{-1}\pi,d)$ for the unitary group are dominated by the terms where $\sigma = \pi$, with the other Weingarten functions suppressed by a factor of $\bigO{1/d}$. The approximation error incurred by keeping only the dominant Weingarten term $\text{Wg}(1_k,d)$ can be bounded as follows:
\begin{lemma}[Large-$d$ approximation of Weingarten sums]
\label{lemma:largedapprox}
For any function $a_{\sigma \pi}$ of the permutation elements $\sigma,\pi$ of the symmetric group $S_k$ with $k^2 < d$, the error of the leading order approximation is
    \begin{equation}
    \normalfont
    \left| \sum_{\pi \in S_k} \text{Wg}(\sigma^{-1}\pi,d) a_{\sigma \pi} - \text{Wg}(1_k,d) a_{\sigma\sigma}\right| = \bigO{ \frac{k^2}{d^{k+1}} \max_{\sigma \neq \pi} |a_{\sigma \pi}|}.
\label{eq:wg_approx_error}
\end{equation}
\end{lemma}
\begin{proof}
    The approximation error in the LHS of Eq.~\eqref{eq:wg_approx_error} can be bounded by
    \begin{equation}
    \begin{split}
        \text{LHS} &= \left| \sum_{\substack{\pi \in S_k \\ \pi \neq \sigma}} \text{Wg}(\sigma^{-1}\pi,d) a_{\sigma \pi} \right| \\&\leq  \sum_{\substack{\pi \in S_k \\ \pi \neq \sigma}}\left| \text{Wg}(\sigma^{-1}\pi,d) \right| \max_{\sigma \neq \pi}|a_{\sigma \pi}| \\&= \left( \sum_{\pi \in S_k} \left|\text{Wg}(\sigma^{-1}\pi,d)\right| - \left|\text{Wg}(1_k,d)\right| \right)  \max_{\sigma \neq \pi}|a_{\sigma \pi}| \\&= \left| \frac{(d-k)!}{d!} - \text{Wg}(1_k,d)\right| \max_{\sigma \neq \pi}|a_{\sigma \pi}| \\&= \bigO{\frac{k^2}{d^{k+1}} \max_{\sigma \neq \pi}|a_{\sigma \pi}|}
    \end{split}
    \end{equation}
where we used~\cite{aharonov2022quantum}
\begin{equation}
    \sum_{\pi \in S_k} \left|\text{Wg}(\sigma^{-1}\pi,d)\right| = \frac{(d-k)!}{d!} = \frac{1}{d^k}\left(1 + \frac{k(k-1)}{2d} + \bigO{\frac{k^4}{d^2}}\right)
\end{equation}
and~\cite{collins2017weingarten}
\begin{equation}
    \text{Wg}(1_k,d) = \frac{1}{d^k}\left(1 + \bigO{\frac{k^{7/2}}{d^2}}\right).
\label{eq:wg_diag_asymptotics}
\end{equation}
\end{proof}
While Lemma~\ref{lemma:largedapprox} is derived for Weingarten functions for the unitary group, it can be adapted for other classical compact groups such as the orthogonal and sympletic groups~\cite{collins2022weingarten}.

\section{Quantum randomness of the projected ensemble}
Here, we provide calculation details for the quantum randomness of the classically-enhanced projected ensemble in Theorem 1 of the main text,
\begin{equation}
    \rho^{(k)} = \mathbb{E}_{x \in \mathcal{E}_{\text{init}}} \sum_{z \in \{0,1\}^{N_B}} p_x(z) \parens{\ket{\psi_x(z)}\bra{\psi_x(z)}}^{\otimes k} = \sum_{\substack{x \in \{0,1\}^{N_A + N_B} \\ z \in \{0,1\}^{N_B}}} q(x) p_x(z) \parens{\ket{\psi_x(z)}\bra{\psi_x(z)}}^{\otimes k}.
\end{equation}
The initial states $\ket{x}$ are drawn from $\mathcal{E}_{\text{init}}$ with probability distribution $q(x)$. 
For notational convenience, we denote the Hilbert space dimensions of subsystems $A$ and $B$ by $d_A = 2^{N_A}$ and $d_B = 2^{N_B}$ respectively. $d = d_A d_B$ is the Hilbert space dimension of the full system $AB$.

As introduced in Appendix~\ref{app:quantifying_qrandom}, we can quantify the quantum randomness via the frame potential
\begin{equation}
    F^{(k)} = \mathbb{E}_{x,x^\prime \in \mathcal{E}_{\text{init}}} \sum_{z,z^\prime \in \{0,1\}^{N_B}} \frac{|\braket{\Phi_x|(I_A \otimes \ket{z}\bra{z^\prime}_B)|\Phi_{x^\prime}}|^{2k}}{\braket{\Phi_x|(I_A \otimes \ket{z}\bra{z}_B)|\Phi_{x}}^{k-1} \braket{\Phi_{x^\prime}|(I_A \otimes \ket{z^\prime}\bra{z^\prime}_B)|\Phi_{x^\prime}}^{k-1} },
\end{equation}
where the so-called generator states $\ket{\Phi_x} = U\ket{x}$ are evolved from the initial states $\ket{x}$ by a common unitary $U$. We consider the average behavior of $F^{(k)}$ by averaging $U$ over the Haar measure on $\mathcal{U}(d_A d_B)$, giving $\mathbb{E}_{\text{Haar}} F^{(k)}$. Using Eq.~\eqref{eq:delta2}, we have
\begin{equation}
    \Delta_{\text{rms}}^{(k)} = \left( \frac{\mathbb{E}_{\text{Haar}} F^{(k)}}{F_{\text{Haar}}^{(k)}(d_A)} - 1 \right)^{1/2},
\label{eq:HSdist_rms}
\end{equation}
where $\Delta_{\text{rms}}^{(k)}$ is the root-mean-square value of the normalized Hilbert-Schmidt distance $\Delta_2^{(k)}$ to the Haar ensemble. By Jensen's inequality, this gives an upper bound to the average value of $\Delta_2^{(k)}$
\begin{equation}
    \mathbb{E}_{\text{Haar}} \Delta_2^{(k)} \leq \Delta_{\text{rms}}^{(k)},
\end{equation}
which is in turn an upper bound for the average trace distance $\mathbb{E}_\text{Haar} \Delta_1^{(k)}$, using Proposition~\ref{prop:dist_inequality}.
Similarly, from Eq.~\eqref{eq:delta2}, the normalized Hilbert-Schmidt distance between the ensemble of initial states $\mathcal{E}_{\text{init}}$ and the Haar ensemble is given by
\begin{equation}
    \Delta_{\text{init}}^{(k)} \equiv \left( \frac{F_{\text{init}}^{(k)}}{F_{\text{Haar}}^{(k)}(d)} - 1 \right)^{1/2} = \left(\frac{ \E_{x,x^\prime \in \mathcal{E}_\text{init}} \left[ |\braket{x|x^\prime}|^{2} \right] }{F_{\text{Haar}}^{(k)}(d)} - 1\right)^{1/2},
\label{eq:HSdist_init}
\end{equation}
\change{where the Haar ensemble is defined on the full system $AB$. $\Delta_{\text{init}}^{(1)}$ measures how close the initial states are to a $1$-design. In the setting discussed in the main text, in which the initial states are computational basis states $\ket{x}$ sampled randomly from the probability distribution $q(x)$, $\Delta_{\text{init}}^{(1)}$ can be expressed as
\begin{equation}
    \Delta_{\text{init}}^{(1)} = \sqrt{2^{N_A + N_B - S_c} - 1},
\end{equation}
where $S_c$ is the R\'enyi 2-entropy of $q(x)$. For example, if $q(x)$ is the uniform distribution over all computational basis states, $S_c = N_A + N_B$ (i.e., $q(x)$ has maximum entropy) and therefore $\Delta_{\text{init}}^{(1)} = 0$.} Our goal is to derive an analytical expression for $\Delta_{\text{rms}}^{(k)}$ in terms of $\Delta_{\text{init}}^{(k)}$. Before tackling the problem for general $k \in \mathbb{N}$, as a warm-up we consider the special case of $k=1$ which corresponds to (regular) thermalization. As we will see, the results for $k=1$ and $k>1$ are qualitatively very different.
\subsection{$k=1$}
The average frame potential is
\begin{equation}
    \E_{\text{Haar}} F^{(1)} = \E_{x,x^\prime \in \mathcal{E}_{\text{init}}} \sum_{z,z^\prime \in \{0,1\}^{N_B}} \text{Tr} \left\{ \bra{z^\prime,z}_B \E_{U \sim \text{Haar}}\left[U^{\otimes 2} \ket{\varphi_{x^\prime},\varphi_{x}}\bra{\varphi_x,\varphi_{x^\prime}} U^{\dag \otimes 2} \right] \ket{z,z^\prime}_B \right\}.
\end{equation}
Using the identity for the Haar average,
\begin{equation}
    \E_{U \sim \text{Haar}(d)} \sparens{U^{\otimes 2} A U^{\dag \otimes 2}} = \frac{1}{d^2 - 1} \left[ \left( \text{Tr}(A) - \frac{1}{d} \text{Tr}(AS) \right) I + \left( \text{Tr}(AS) - \frac{1}{d} \text{Tr}(A) \right) S \right]
\end{equation}
where $I$ is the identity operator and $S$ is the swap operator, we have
\begin{equation}
    \E_{\text{Haar}} F^{(1)} = \frac{d_A(d_B^2 - 1)}{d_A^2 d_B^2 - 1} + \frac{d_B (d_A^2 - 1)}{d_A^2 d_B^2 - 1} \E_{x,x^\prime \in \mathcal{E}_{\text{init}}} \left[ |\braket{x|x^\prime}|^{2} \right].
\end{equation}
Since the first moment $\rho^{(1)}$ of the projected ensemble is the reduced density matrix on $A$, $F^{(1)}$ is simply the subsystem purity of $A$. Using Eqs.~\eqref{eq:HSdist_rms} and~\eqref{eq:HSdist_init}, we obtain
\begin{equation}
    \Delta_{\text{rms}}^{(1)} = \Delta_{\text{init}}^{(1)} \left( \frac{d_A^2 - 1}{d_A^2 d_B^2 - 1} \right)^{1/2}.
\label{eq:delta_rms_k1}
\end{equation}
This has the physical interpretation that the infinite-temperature thermalization of the subsystem $A$, measured by $\Delta_{\text{rms}}^{(1)}$, comes from two sources: the randomness of the initial states characterized by $\Delta_{\text{init}}^{(1)}$, and the dimension of the bath $d_B$. In the language of $k$-designs, if the ensemble of initial states $\mathcal{E}_{\text{init}}$ form an exact $1$-design, the projected ensemble $\mathcal{E}$ is also an exact $1$-design. This can be easily achieved by choosing, for example, $\mathcal{E}_{\text{init}}$ to be the complete set of computational basis states.

Eq.~\eqref{eq:delta_rms_k1} hints at the possibility of converting classical to quantum randomness, since $\Delta_{\text{init}}^{(1)}$ vanishes exponentially with the R\'{e}nyi 2-entropy $S_c$ of the distribution $q(x)$ for computational basis states. However, the case of $k=1$ is qualitatively different from $k > 1$ in several important aspects:
\begin{enumerate}
    \item As mentioned above, $\rho^{(1)}$ is the reduced density matrix on subsystem $A$. Thus, $\rho^{(1)}$ does not depend on the specific measurement basis on $B$ (or require any measurement at all), as long as the basis is complete. This is not true for $k > 1$, where the behavior of higher moments depends on the choice of measurement basis~\cite{mark2024maximum}. This is especially relevant when $U$ is a Hamiltonian evolution where energy is conserved, since the measurement outcome on $B$ can potentially be highly correlated with the projected state on $A$.
    \item An \textit{exact} $1$-design can be formed trivially by initializing the qubits on $AB$ with uniformly random computational basis states, even without any bath qubits ($d_B = 1$). For $k > 1$, the bath qubits are necessary to form \textit{approximate} $k$-designs for any finite bath size, even with the injection of classical randomness. Physically, this reflects the fact that $1$-designs do not require quantum resources such as coherence or entanglement, while entanglement is necessary for $k$-designs with $k > 1$.
    Although it is, in principle, possible for the projected ensemble to form exact $k$-designs, this requires a highly fine-tuned unitary $U$ and is generically not accessible by physical Hamiltonian evolution. 
\end{enumerate}

\subsection{$k > 1$}
\label{sec:sm_derivation}
In the original projected ensembles protocol~\cite{Cotler2023emergent,choi2023preparing} with $|\mathcal{E}_{\text{init}}| = 1$ (i.e., fixed initial state), it was demonstrated that quantum randomness can be generated from a single fixed unitary $U$, by measuring a sufficiently large number of bath qubits. Here, the quantum randomness stems from the intrinsic quantum-mechanical randomness of the bath measurement outcomes. As such, the amount of quantum randomness achievable is limited by the number of bath qubits $N_B$.

By deriving an analogous formula to Eq.~\eqref{eq:delta_rms_k1} for $k > 1$, we show that injecting classical randomness via the initial state distribution $q(x)$ results in enhanced quantum randomness. In fact, Theorem 1 in the main text implies that the conversion of classical to quantum randomness in our protocol is optimal: injecting an additional bit of classical randomness is equivalent to adding an extra bath qubit in the original protocol. We now prove Theorem 1.

\begin{proof}[Proof of Theorem 1]
We employ the replica trick~\cite{Ippoliti2022solvable} by writing
\begin{equation}
    \E_{\text{Haar}}F^{(k)} = \E_{x,x^\prime \in \mathcal{E}_{\text{init}}} G^{(1-k,k)}_{x,x^\prime} = \sum_{x} q(x)^2 G_{x,x}^{(1-k,k)} + \sum_{x\neq x^\prime} q(x)q(x^\prime) G_{x,x^\prime}^{(1-k,k)},
\label{eq:avghaarfk_replica}
\end{equation}
where
\begin{equation}
\begin{split}
    G^{(m,k)}_{x,x^\prime} = \E_{U \sim \text{Haar}} \sum_{z,z^\prime \in \{0,1\}^{N_B}} &\Bigg[\parens{\braket{x|U^\dag (I_A \otimes \ket{z}\bra{z}_B) U|x} \braket{x^\prime|U^\dag (I_A \otimes \ket{z^\prime}\bra{z^\prime}_B) U|x^\prime}}^m \\&\times \left|\braket{x|U^\dag (I_A \otimes \ket{z}\bra{z^\prime}_B) U|x^\prime}\right|^{2k}\Bigg]
\end{split}
\end{equation}
can be regarded as the analytical continuation of the frame potential~\cite{claeys2022emergent}. Essentially, in the replica trick, we first assume $m \in \mathbb{N}$ and then take the limit $m \to 1-k$ at the end, after $G^{(m,k)}_{x,x^\prime}$ is evaluated analytically. Since $G^{(m,k)}_{x,x^\prime}$ contains a polynomial function of degree $2(m+k)$ in $U$ and $U^\dag$, the Haar average can be evaluated using Weingarten calculus~\cite{kostenberger2021weingarten,collins2022weingarten}, giving
\begin{equation}
    G^{(m,k)}_{x,x^\prime} = \sum_{z,z^\prime} \sum_{\sigma,\pi \in S_{2(k+m)}} \text{Wg}(\sigma^{-1}\pi,d) \braket{x^{m} x^{k} x^{\prime k} x^{\prime m}|\hat{\pi}^{\dag}|x^{ m} x^{\prime  k} x^{ k} x^{\prime  m}} d_A^{\#\text{cycles}(\sigma)} \braket{z^m z^{\prime k} z^k z^{\prime m}|\hat{\sigma}_B|z^m z^k z^{\prime k} z^{\prime m}}.
\label{eq:Gmk}
\end{equation}
In the above expression, $\sigma$ and $\pi$ are permutation elements of the symmetric group $S_{2(k+m)}$ with the corresponding permutation unitary operators denoted by hats. The permutation operators permute the $2(m+k)$ tensor factors. The subscript on $\hat{\sigma}_B$ indicates that the permutation operator $\hat{\sigma}$ acts on subsystem $B$. $\#\text{cycles}(\sigma)$ counts the number of cycles in the permutation $\sigma$. We also introduce the shorthand notation $\ket{x^m x^{\prime k} x^k x^{\prime m}} \equiv \ket{x}^{\otimes m} \ket{x^\prime}^{\otimes k} \ket{x}^{\otimes k} \ket{x^\prime}^{\otimes m}$ to denote the $2(m+k)$-fold tensor product. The case of $x = x^\prime$ was previously studied in Ref.~\cite{Ippoliti2022solvable}, with
\begin{equation}
    G_{x,x}^{(1-k,k)} = \frac{d_A + 1 + d_A(d_B-1)F_\text{Haar}^{(k)}(d_A)}{d_A d_B + 1}.
\label{eq:framepot_singleinit}
\end{equation}
Therefore, we will focus on the case of $x \neq x^\prime$ where $\braket{x^\prime|x} = 0$ since the initial states are mutually orthonormal. The general scenario of non-orthogonal initial states is discussed later in Appendix~\ref{sec:nonorthogonal}. 

Splitting the sum over $z$ and $z^\prime$ into terms where $z = z^\prime$ and $z \neq z^\prime$, we have
\begin{equation}
\begin{split}
G_{x,x^\prime}^{(m,k)} &= d_B \sum_{\substack{\sigma \in S_{2(k+m)} \\ \pi_1,\pi_2 \in S_{k+m}}} \text{Wg}(\sigma^{-1}(\tau \pi_1\pi_2),d) d_A^{\#\text{cycles}(\sigma)} \\&+ d_B(d_B-1) \sum_{\substack{\pi_1,\pi_2,\sigma_1,\sigma_2 \in S_{k+m}}} \text{Wg}((\tau\sigma_1\sigma_2)^{-1}(\tau\pi_1\pi_2),d) d_A^{\#\text{cycles}(\tau\pi_1\pi_2)},
\end{split}
\end{equation}
using the orthogonality of the initial states and the measurement basis to keep only the non-vanishing contributions. $\hat{\tau}$ is the permutation operator that swaps the two $k$-fold tensor factors and acts identically on the remaining $2m$ factors, e.g., $\hat{\tau} \ket{x^m x^{\prime k} x^k x^{\prime m}} = \ket{x^m x^k x^{\prime k} x^{\prime m}}$. $\hat{\sigma}_1$ and $\hat{\sigma}_2$ acts only on the first and second $m+k$ factors respectively, i.e., $\widehat{\sigma_1 \sigma_2} = \hat{\sigma}_1 \otimes \hat{\sigma}_2$. 

To proceed, we make use of Lemma~\ref{lemma:largedapprox} to approximate the Weingarten sums for large $d > k^2$. This gives
\begin{equation}
\begin{split}
    G_{x,x^\prime}^{(m,k)} &= d_B \sum_{\pi_1,\pi_2 \in S_{k+m}} \left[ \text{Wg}(1_{2(k+m)},d) d_A^{\#\text{cycles}(\tau\pi_1\pi_2)} + \bigO{\frac{(k+m)^2 d_A^{2(k+m)}}{d^{2(k+m)+1}} }\right] \\&+ d_B(d_B-1) \sum_{\sigma_1,\sigma_2 \in S_{k+m}} d_A^{\#\text{cycles}(\tau\sigma_1\sigma_2)} \left(\text{Wg}(1_{2(k+m)},d)+\bigO{\frac{(k+m)^2}{d^{2(k+m)+1}}}\right) \\&= d_B \left[\text{Wg}(1_{2(k+m)},d) \frac{(d_A+k+m-1)!^2}{d_A-1)!^2}F_{\text{Haar}}^{(k)}(d_A) + \bigO{\frac{(k+m)^2(k+m)!^2 d_A^{2(k+m)-1}}{d^{2(k+m)+1}}}\right] \\&+ d_B(d_B-1) \frac{(d_A+k+m-1)!^2}{(d_A-1)!^2} F_{\text{Haar}}^{(k)}(d_A) \left[\text{Wg}(1_{2(k+m)},d)+\bigO{\frac{(k+m)^2 (k+m)!^2}{d^{2(k+m)+1}}}\right]
\end{split}
\label{eq:Gxxp_mk}
\end{equation}
where in the last line we used the identity~\cite{Ippoliti2022solvable}
\begin{equation}
    \sum_{\sigma_1,\sigma_2 \in S_{k+m}} d_A^{\#\text{cycles}(\tau \sigma_1\sigma_2)} = \frac{(d_A + k+m-1)!^2}{(d_A - 1)!^2} F_{\text{Haar}}^{(k)}(d_A).
\label{eq:fhaar_identity}
\end{equation}
Note that for $m \geq 1-k$, the final expression in Eq.~\eqref{eq:fhaar_identity} can be extended to hold for real $m$, by replacing the factorials with the appropriate Gamma functions. Analyticity of $G_{x,x^\prime}^{(m,k)}$ in this extended domain then justifies setting $m = 1-k$, which yields
\begin{equation}
\begin{split}
    G_{x,x^\prime}^{(1-k,k)} &= d_B \left[\text{Wg}(1_2,d) d_A^2 F_{\text{Haar}}^{(k)}(d_A) + \bigO{\frac{d_A}{d^3}}\right] + d_B (d_B-1) d_A^2 F_{\text{Haar}}^{(k)}(d_A) \left[\text{Wg}(1_2,d)+\bigO{\frac{1}{d^3}}\right] \\&= \left[d^2 \text{Wg}(1_2,d) + \bigO{\frac{1}{d}}\right] F_{\text{Haar}}^{(k)}(d_A) + \bigO{\frac{1}{d_A^2 d_B^2}} \\&= \left[1 + \bigO{\frac{1}{d_A d_B}} \right] F_{\text{Haar}}^{(k)}(d_A) + \bigO{\frac{1}{d_A^2 d_B^2}}
\end{split}
\end{equation}
using $\text{Wg}(1_2,d) = 1/(d^2-1)$. Finally, substituting the results into Eq.~\eqref{eq:avghaarfk_replica},
 we arrive at (to leading order as $d_A,d_B \to \infty$)
\begin{equation}
    \E_{\text{Haar}}F^{(k)} = F_{\text{Haar}}^{(k)}(d_A) + \frac{1}{d_B} \sum_x q(x)^2 = F_{\text{Haar}}^{(k)}(d_A) + \frac{1+(\Delta_{\text{init}}^{(1)})^2}{d_A d_B^2}, 
\label{eq:framepot_reinit}
\end{equation}
where we used $(\Delta_\text{init}^{(1)})^2 = d_A d_B\sum_x q(x)^2 - 1$. 
Thus, the deviation between the average frame potential $\E_{\text{Haar}}F^{(k)}$ and the frame potential $F_{\text{Haar}}^{(k)}$ of the Haar ensemble depends on how close the initial state ensemble $\mathcal{E}_{\text{init}}$ forms a $1$-design. Substituting this into Eq.~\eqref{eq:HSdist_rms}, and using $\sum_x q(x)^2 = 2^{-S_c}$ yields the desired result.
\end{proof}
As a consistency check, let us consider the limit where $q(x) = \delta_{x,0}$, i.e., a fixed initial state which we set as $\ket{0}$ without loss of generality. Then, $1+(\Delta_\text{init}^{(1) })^2 = d_A d_B$ and we get
\begin{equation}
    \E_\text{Haar}F^{(k)} \approx F_{\text{Haar}}^{(k)}(d_A) + \frac{1}{d_B}
\end{equation}
which is consistent with previous result~\cite{Ippoliti2022solvable}. We also note that in the non-orthogonal case of $\braket{x^\prime|x} \neq 0$ for $x \neq x^\prime$, there will be an additive positive correction to Eq.~\eqref{eq:framepot_reinit} that depends on $|\braket{x^\prime|x}|^2$. This can be traced all the way back to Eq.~\eqref{eq:Gmk} where the matrix element contributes
\begin{equation}
    \braket{x^{m} x^{k} x^{\prime k} x^{\prime m}|\hat{\pi}^{\dag}|x^{ m} x^{\prime  k} x^{ k} x^{\prime  m}} = \text{some even power of } |\braket{x^\prime|x}|^2.
\end{equation}
Physically, this means that for a fixed distribution $q(x)$, quantum randomness is optimized for orthonormal initial states (such as computational basis states) in order to maximize entropy.

\subsection{Example: Higher-order $k$-designs with classical randomness}
As an explicit demonstration of the analytical result in Theorem 1, we show that injecting classical randomness can lead to the formation of higher-order $k$-designs which are not possible in the original protocol of~\cite{Cotler2023emergent,choi2023preparing} where $S_c = 0$.

For simplicity, let us consider $N_A = N_B$. The trace distance $\Delta_1^{(1)}/2$ between the first moment of the projected ensemble (i.e., the reduced density matrix on $A$) and the maximally mixed state $I_A/d_A$ is typically not small~\cite{cramer2012thermalization}, and the ensemble does not converge to a $1$-design. Injecting classical randomness by initializing the qubits in random computational basis states on $A$ can trivially produce a $1$-design. More interestingly, from Theorem 1, the projected ensemble can in fact form an approximate $2$-design, with the Haar average trace distance bounded by
\begin{equation}
    \E_\text{Haar} \Delta_1^{(2)} \leq \E_\text{Haar} \Delta_2^{(2)} \leq \Delta_\text{rms}^{(2)} = \frac{1}{2^{S_c - N_A + 1}}
\end{equation}
using Proposition~\ref{prop:dist_inequality}. Therefore, for $S_c - N_A \gg 1$, the projected ensemble is close to the Haar ensemble, up to the second moment. In fact, to generate an approximate $2$-design with an $o(1)$ trace distance, our results suggest that the bath can be as small as $N_B = N_A/2 + \omega(1)$, with $S_c = N_A + N_B$ bits of classical entropy. 

\begin{figure}
\centering
\subfloat{%
  \includegraphics[width=0.7\linewidth]{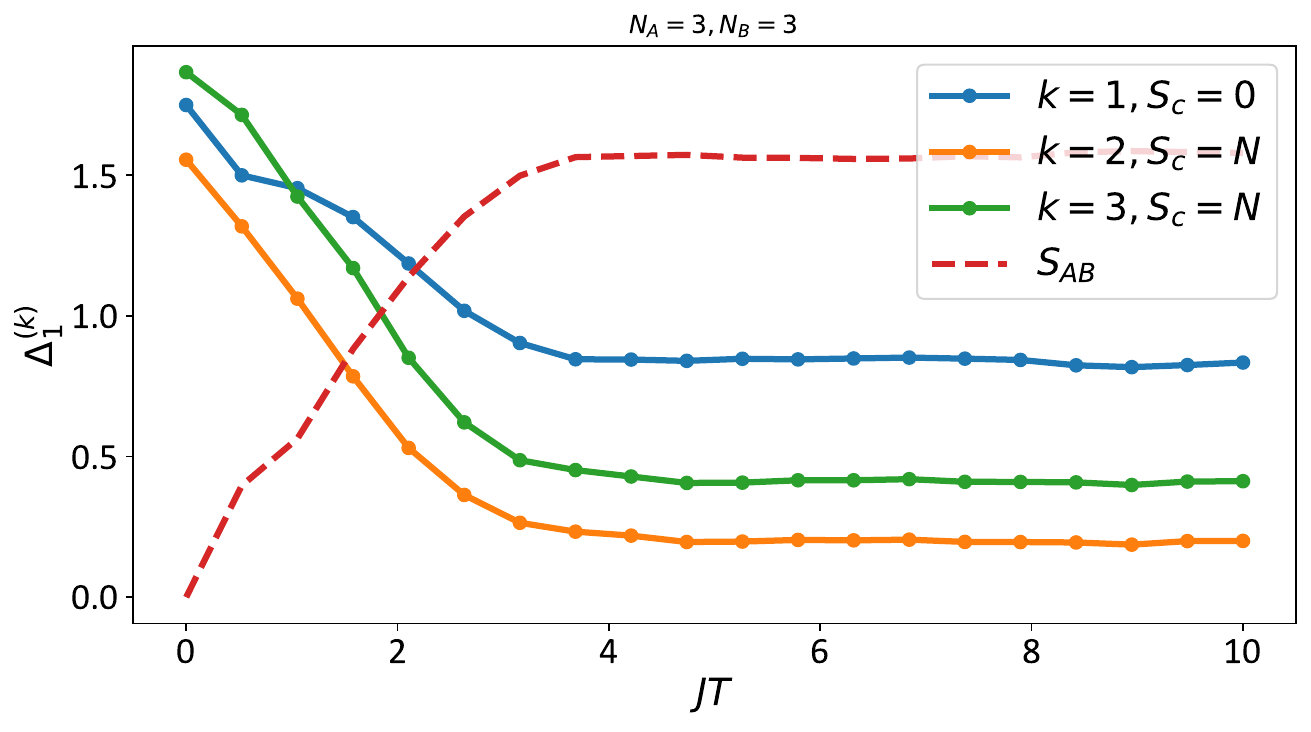}%
}
\caption{Distance $\Delta_1^{(k)}$ between the $k$-th moment of the projected ensemble and the Haar ensemble, measured by the Schatten-$1$ norm. Numerical results are obtained for the 1D mixed field Ising model (open boundary conditions) with $N_A = N_B = 3$, parameters $(h_x,h_y,J) = (0.809,0.9045,1)$, and a quench time $T$. The red dashed line denotes the entanglement entropy $S_{AB}$ between subsystems $A$ and $B$, for a fixed initial state (i.e., $S_c = 0$). The data for $k=1$ (blue) and $S_{AB}$ (red) are averaged over all $2^{N}$ initial computational basis states, where $N = N_A + N_B = 6$.}
 \label{fig:halfchain}
\end{figure}

More generally, consider $N_B = (k-1) N_A$, for $k \geq 2$. The projected ensemble with $S_c = 0$ typically forms an approximate $(k-2)$ design, but not a $k$-design (the case of $k-1$ is marginal). This can be understood from the cardinality bound of $k$-designs~\cite{roberts2017chaos}, or more explicitly, by applying the Fannes-Audenaert inequality~\cite{audenaert2007sharp} to derive
\begin{equation}
    \frac{1}{2}\Delta_1^{(k)} \geq 1 - \frac{N_B}{k N_A - \log_2 (k!)}
\end{equation}
valid for $N_B \ll k N_A - \log_2(k!)$. On the other hand, when $S_c = N_A + N_B = k N_A$, Theorem 1 implies the possibility of forming approximate $2(k-1)$ designs, since
\begin{equation}
    \E_\text{Haar} \Delta_1^{(k^\prime)} \leq \E_\text{Haar} \Delta_2^{(k^\prime)} \leq \Delta_\text{rms}^{(k^\prime)} = \frac{1}{k^\prime! 2^{(2k-1-k^\prime)N_A}},
\end{equation}
which vanishes exponentially with $N_A$ if $k^\prime = 2(k-1)$. This is roughly a factor of $2$ increase in the maximum order of designs achievable.  

In Fig.~\ref{fig:halfchain}, we show numerical results for the 1D mixed field Ising model with $N_A = N_B = 3$ and open boundary conditions. At the thermalization time where the entanglement entropy $S_{AB}$ between the system and bath saturates, the projected ensemble is close to a 2-design. We leave it as future work to investigate if the formation of higher-order designs with $k > 1$ has a longer time scale than usual thermalization.

\subsection{Non-orthogonal initial states}
\label{sec:nonorthogonal}
 Here, we derive the correction terms to Eq.~\eqref{eq:framepot_reinit} when the initial states are not mutually orthonormal. Equation~\eqref{eq:Gxxp_mk} becomes (when $x \neq x^\prime$)
\begin{equation}
    G_{x,x^\prime}^{(m,k)} =~\eqref{eq:Gxxp_mk} + d_B r_1 + d_B(d_B-1) r_2,
\end{equation}
where the first term corresponds to the case in which $\braket{x^\prime|x} = \delta_{x,x^\prime}$, while the corrections are
\begin{equation}
    r_1 \equiv \sum_{\substack{\sigma,\pi \in S_{2(m+k)} \\ \pi \neq \tau \pi_1 \pi_2}} \text{Wg}(\sigma^{-1}\pi,d) d_A^{\#\text{cycles}(\sigma)} \braket{x^{m} x^{k} x^{\prime k} x^{\prime m}|\hat{\pi}^{\dag}|x^{ m} x^{\prime  k} x^{ k} x^{\prime  m}}
\end{equation}
and
\begin{equation}
    r_2 \equiv \sum_{\substack{\pi \in S_{2(m+k)} \\ \sigma_1,\sigma_2 \in S_{m+k} \\ \pi \neq \tau \pi_1 \pi_2}} \text{Wg}((\tau \sigma_1 \sigma_2)^{-1}\pi,d) d_A^{\# \text{cycles}(\tau \sigma_1 \sigma_2)} \braket{x^{m} x^{k} x^{\prime k} x^{\prime m}|\hat{\pi}^{\dag}|x^{ m} x^{\prime  k} x^{ k} x^{\prime  m}}.
\end{equation}
The notation $\pi \neq \tau \pi_1 \pi_2$ indicates that the sum over permutations $\pi$ only includes permutations which are not of the form $\tau \pi_1 \pi_2$. We now establish bounds on the correction terms. First,
\begin{equation}
\begin{split}
    r_1 &\leq |\braket{x^\prime|x}|^2 \sum_{\substack{\sigma,\pi \in S_{2(m+k)} \\ \pi \neq \tau \pi_1 \pi_2}} \text{Wg}(\sigma^{-1}\pi,d) d_A^{\#\text{cycles}(\sigma)} \\&= |\braket{x^\prime|x}|^2 \sum_{\substack{\pi \in S_{2(m+k)} \\ \pi \neq \tau \pi_1 \pi_2}} \sparens{\text{Wg}(1_{2(m+k)},d) d_A^{\#\text{cycles}(\pi)} + \bigO{\frac{(m+k)^2 d_A^{2(m+k)}}{d^{2(m+k)+1}}}} \\&\leq |\braket{x^\prime|x}|^2 \sparens{\text{Wg}(1_{2(m+k)},d) \frac{(d_A + 2(m+k)-1)!}{(d_A-1)!} + \bigO{\frac{[2(m+k)]!^2(m+k)^2 d_A^{2(m+k)}}{d^{2(m+k)+1}}}} \\&\to \bigO{\frac{|\braket{x^\prime|x}|^2}{d_B^2}}, 
\end{split}
\end{equation}
setting $m = 1-k$ in the final step. Next,
\begin{equation}
\begin{split}
    r_2 &\leq |\braket{x^\prime|x}|^2 \sum_{\substack{\pi \in S_{2(m+k)} \\ \sigma_1,\sigma_2 \in S_{m+k} \\ \pi \neq \tau \pi_1 \pi_2}} \text{Wg}((\tau \sigma_1 \sigma_2)^{-1}\pi,d) d_A^{\# \text{cycles}(\tau \sigma_1 \sigma_2)} \\&= |\braket{x^\prime|x}|^2 \sum_{\sigma_1,\sigma_2 \in S_{m+k}} d_A^{\# \text{cycles}(\tau \sigma_1 \sigma_2)} \sum_{\substack{\pi \in S_{2(m+k)} \\ \pi \neq \tau \pi_1 \pi_2}} \text{Wg}((\tau \sigma_1 \sigma_2)^{-1}\pi,d) \\&\leq |\braket{x^\prime|x}|^2 \sum_{\sigma_1,\sigma_2 \in S_{m+k}} d_A^{\# \text{cycles}(\tau \sigma_1 \sigma_2)} \sum_{\pi \in S_{2(m+k)}} |\text{Wg}((\tau \sigma_1 \sigma_2)^{-1}\pi,d)| \\&= |\braket{x^\prime|x}|^2 \sum_{\sigma_1,\sigma_2 \in S_{m+k}} d_A^{\# \text{cycles}(\tau \sigma_1 \sigma_2)} \frac{(d-2(m+k))!}{d!} \\&= |\braket{x^\prime|x}|^2 \sparens{\frac{(d_A+m+k-1)!}{d_A-1)!}}^2 F_{\text{Haar}}^{(k)}(d_A) \frac{(d-2(m+k))!}{d!} \\&\to \bigO{\frac{|\braket{x^\prime|x}|^2}{d_B^2} F_{\text{Haar}}^{(k)}(d_A)},
\end{split}
\end{equation}
setting $m = 1-k$ in the final step. Using Eq.~\eqref{eq:avghaarfk_replica}, we have the correction to the average frame potential in Eq.~\eqref{eq:framepot_reinit} as
\begin{equation}
    \E_{\text{Haar}}F^{(k)} = \sparens{1+\bigO{\overline{|\braket{x^\prime|x}|^2}}} F_{\text{Haar}}^{(k)}(d_A) + \frac{1}{d_B} \sparens{\sum_x q(x)^2 + \bigO{\overline{|\braket{x^\prime|x}|^2}}},
\end{equation}
where
\begin{equation}
    \overline{|\braket{x^\prime|x}|^2} = \E_{x,x^\prime} \sparens{|\braket{x^\prime|x}|^2 \bigg| x\neq x^\prime} = \sum_{x \neq x^\prime} q(x) q(x^\prime) |\braket{x^\prime|x}|^2
\end{equation}
is the expectation value of the squared overlaps $|\braket{x^\prime|x}|^2$, conditioned on $x \neq x^\prime$. Thus, Theorem 1 is valid even for non-orthogonal initial states, if
\begin{equation}
    \overline{|\braket{x^\prime|x}|^2} \ll \max\left\{\frac{1}{2^{S_c}},\frac{1}{k! 2^{S_c + N_B - k N_A}} \right\}.
\label{eq:overlapbound}
\end{equation}

\subsection{Robustness of numerical results}
Here, we show that our numerical results in Fig. 1(c) of the main text do not depend strongly on the Hamiltonian parameters, which is consistent with the expectation that quantum many-body chaos is robust. As explained in Ref.~\cite{kim2013ballistic}, $H_0$ remains chaotic (in the Wigner-Dyson sense) as long as $h_x$ and $h_y$ are comparable to $J$. The choice of parameters has an influence on the short-time dynamics, which does not scale with the system size and thus does not affect the asymptotic properties. Consequently, we expect deep thermalization to hold for a wide range of parameters. To this end, we numerically study a random Hamiltonian of the form
\begin{equation}
    H = H_0 + H_1,
\end{equation}
where $H_0$ is the (deterministic) mixed-field Ising Hamiltonian given by Eq. (4) in the main text, and
\begin{equation}
    H_1 = \sum_{j=1}^{N_A + N_B} (\eta_x X_j + \eta_y Y_j),
\end{equation}
with $\eta_x, \eta_y$ independent uniform random variables $\in [-0.5,0.5]$ in units of the interaction strength $J$. The numerical results for $\Delta^{(k)}$ are shown in Fig.~\ref{fig:randomMFIM}, analogous to Fig. 1(c) in the main text, averaged over 10 random realizations of $H$. Notice that the standard deviation of $\Delta^{(k)}$, represented by the error bars, is very small. This suggests that the optimal conversion of classical to quantum randomness is robust and holds generically for chaotic many-body Hamiltonians.
\begin{figure}
\centering
\includegraphics[width=0.6\textwidth]{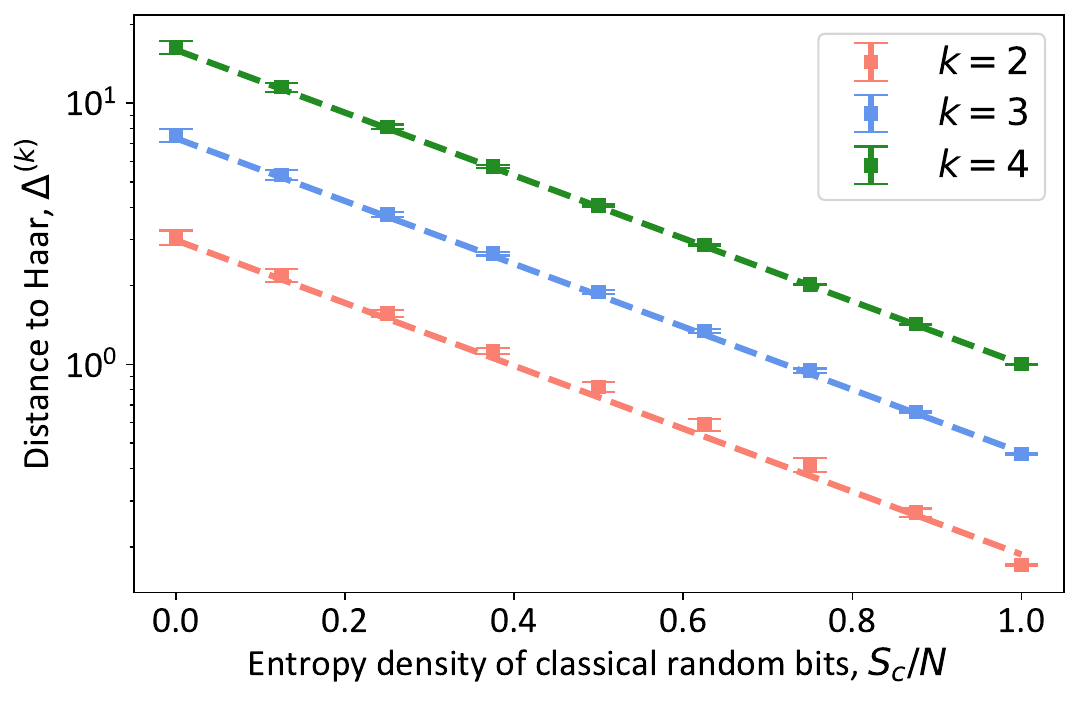}
\caption{Normalized Hilbert-Schmidt distance $\Delta^{(k)}$ between the $k$-th moment of the projected ensemble $\mathcal{E}$ and the Haar ensemble, against the entropy density $S_{\text{c}}/N$ of the classical distribution $q(x)$, with $N_A = N_B = 4$. The points are obtained numerically by evolving the initial state with the random Hamiltonian $H = H_0 + H_1$ for a time $JT = 10^3$. The distance is averaged over 10 realizations of $H$, with the standard deviation of $\Delta^{(k)}$ represented by error bars. The dashed lines denote the analytical root-mean-square distance $\Delta_{\text{rms}}^{(k)}$ when $U$ is a Haar random unitary.}
 \label{fig:randomMFIM}
\end{figure}

\section{Additional results on Hamiltonian disorder}

\subsection{Estimates on the timescale}
For a $D$-dimensional Hamiltonian with geometrically local interactions, the timescale for the $N_A$ system qubits to thermalize with the bath qubits is typically $\sim N_A^{1/D}$, which can be intuitively understood as the time taken for volume-law entanglement to be built up between subsystems $A$ and $B$. Previous numerical and analytical results~\cite{Cotler2023emergent,Ippoliti2022solvable,chan2024projected} suggest that the timescale for deep thermalization (i.e. to form approximate $k$-designs) does not depend strongly on $k$. For our scheme of converting classical to quantum randomness via Hamiltonian disorder to be feasible, it is imperative to understand the timescales required in terms of the number of qubits. Here, we provide an estimate based on a physical argument.

In the case where the initial state $\ket{\psi_0}$ is fixed and classical randomness is injected via disorder in the Hamiltonian, the condition~\eqref{eq:overlapbound} becomes relevant to achieve optimal conversion of randomness, where $\ket{x}$ and $\ket{x^\prime}$ are interpreted as the time-evolved state corresponding to two independent realizations of disorder. 

The overlap can be written as
\begin{equation}
    F(t) = \left| \braket{\psi_0|e^{i(H_0 + V_2)t} e^{-i(H_0 + V_1)t}|\psi_0}\right|^2 \equiv |\braket{e^{i(H_0 + V_2)t} e^{-i(H_0 + V_1)t}}_0|^2,
\label{eq:echofidelity}
\end{equation}
analogous to the Loschmidt echo, which typically decays exponentially for chaotic dynamics~\cite{gorin2006dynamics}. Here, $V_1$ and $V_2$ are arbitrary Hamiltonians in general, although for our purposes these correspond to independent realizations of the disorder Hamiltonian. Expanding $F(t)$ to order $t^2$, we have
\begin{equation}
    F(t) = 1 - t^2 \sparens{(\delta V_1)^2 + (\delta V_2)^2 - \braket{V_1 V_2}_c - \braket{V_2 V_1}_c} + \bigO{t^3},
\end{equation}
where
\begin{equation}
    (\delta V_j)^2 \equiv \braket{V_j^2}_0 - (\braket{V_j}_0)^2, \quad j = 1,2
\end{equation}
is the variance of $V_j$ for the initial state $\ket{\psi_0}$, and
\begin{equation}
    \braket{V_j V_k}_c \equiv \braket{V_j V_k}_0 - \braket{V_j}_0 \braket{V_k}_0, \quad j,k = 1,2
\end{equation}
is the connected correlation function for the initial state. From our expectation that $F(t)$ decays exponentially, we posit that
\begin{equation}
    F(t) \approx \exp\sparens{-t^2\parens{(\delta V_1)^2 + (\delta V_2)^2 - \braket{V_1 V_2}_c - \braket{V_2 V_1}_c}}.
\end{equation}
To test the validity of this formula, we consider the example where $H_0$ is the 1D mixed-field Ising Hamiltonian (5) in the main text, while $V_1$ and $V_2$ are independent realizations of the disorder Hamiltonian $H_d$ in Eq. (5) of the main text. We write
\begin{equation}
    V_1 = \sum_{i=1}^{N} \xi_i X_i, \quad V_2 = \sum_{i=1}^{N} \xi_i^\prime X_i,
\end{equation}
where $\xi_i,\xi_i^\prime \overset{\text{i.i.d.}}{\sim} \text{Uniform}[-W,W]$. The initial state is chosen to be any computational basis state. Thus, $\braket{V_1}_0 = \braket{V_2}_0 = 0$,
\begin{equation}
    (\delta V_1)^2 = \sum_{i,j=1}^{N} \xi_i \xi_j \braket{X_i X_j}_0 = \sum_{i=1}^{N} \xi_i^2 \approx \frac{1}{3} N W^2,
\end{equation}
and similarly for $(\delta V_2)^2$. The correlation functions are
\begin{equation}
    \braket{V_1 V_2}_c = \braket{V_1 V_2}_0 = \sum_{i,j} \xi_i \xi_j^\prime \braket{X_i X_j}_0 = \sum_i \xi_i \xi_i^\prime \approx 0,
\end{equation}
and also $\braket{V_2 V_1}_c \approx 0$. Thus, for the 1D mixed-field Ising model with disorder, we have
\begin{equation}
    F(t) \approx \exp\parens{-\frac{2}{3}N W^2 t^2}.
\label{eq:echofid_approx}
\end{equation}
This is numerically verified in Fig.~\ref{fig:echo_fidelity} for $N = 12$ qubits. Therefore, the condition~\eqref{eq:overlapbound} becomes (assuming $N_B \geq k N_A$),
\begin{equation}
    \exp\parens{-\frac{2}{3} (N_A + N_B) W^2 T^2} \ll \frac{1}{k! 2^{S_c + N_B - k N_A}}
\end{equation}
which for large $N_B$ and constant $k$ gives $WT \gg 1$, as stated in the main text. Note that this is only the timescale for the evolved states to be sufficiently orthogonal. For deep thermalization, one would further require $JT = \Omega(N_A)$ to generate volume-law entanglement between $A$ and $B$.

\begin{figure}
\centering
\subfloat{%
  \includegraphics[width=0.5\linewidth]{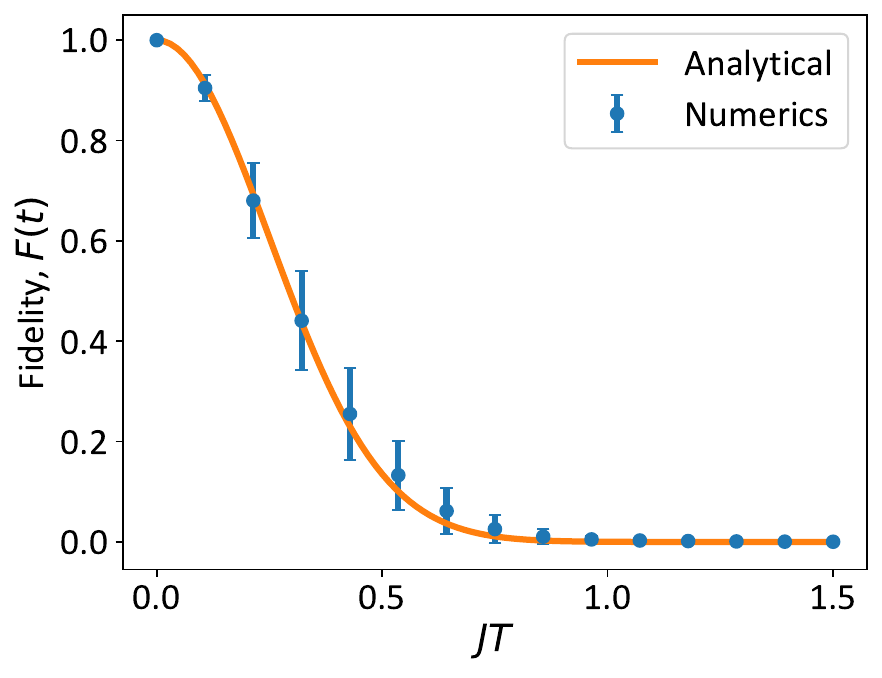}%
}
\caption{Fidelity $F(t)$ between two states evolved under a 1D mixed-field Ising model $H_0$ with different independent realizations of the longitudinal-field disorder Hamiltonian $H_d$ with disorder strength $W/J$. Blue data points are obtained from numerical simulation of $N = 12$ qubits with parameters $(h_x,h_y,J) = (0.809,0.9045,1)$ and $W = 1$, averaged over $30$ random instances. The standard deviation is represented by error bars. Orange line shows the analytical approximation~\eqref{eq:echofid_approx}.}
 \label{fig:echo_fidelity}
\end{figure}

\subsection{Many-body localization}
Here, we study many-body localization of the mixed-field Ising model used for the projected ensemble.
Following Eq. (5) of the main text, we consider the 1D mixed-field Ising Hamiltonian with open boundary conditions, where we use the parameters $\{h_x, h_y, J\} = \{0.8090,0.9045,1\}$:
\begin{equation}
\label{eq:ising_sup}
    H = \sum_{i=1}^{N_A+N_B} (h_x X_i + h_y Y_i + J X_i X_{i+1})  +\sum_{i=1}^{N_A+N_B}\xi_i X_i,
\end{equation}
with Pauli operators $\{X_i,Y_i,Z_i\}$ and
spatially inhomogeneous disorder field $\xi_i\in[-W,W]$ which is sampled randomly from the uniform distribution. 
For $W/J=0$, the Hamiltonian is non-integrable and exhibits chaos~\cite{kim2014testing}. In contrast, for sufficiently large disorder strengths $W/J$, the disorder impedes the propagation of excitations, effectively localizing the dynamics, which is known as many-body localization. 
Such systems typically exhibit a cross-over from chaotic to localized dynamics at a critical disorder $W_\text{c}$~\cite{oganesyan2007localization}.

Numerically, we estimate $W_\text{c}/J\approx0.42$, as shown in Fig.~\ref{fig:gapratio}. To do this, we compute the spectral gap ratio $r$: First, we diagonalize $H$ and determine its eigenvalue spectrum $E_n$, which we sort in ascending order $E_n\leq E_{n+1}$ for $n=1,\dots,2^{N}$, where $N = N_A + N_B$ is the total number of qubits. We compute the gap between neighboring eigenvalues $g_n=E_{n+1}-E_n$, and construct the gap ratio 
\begin{equation}
    r_n=\frac{\text{min}(g_n,g_{n+1})}{\text{max}(g_n,g_{n+1})}\,.
\end{equation}
The distribution of $r_n$ is known to be Poissonian in the localized regime, while following a Wigner-Dyson distribution in the chaotic regime.
The average gap ratio $\langle r\rangle\approx0.53590$ for the Wigner-Dyson distribution of the Gaussian Orthogonal Ensemble (GOE), and $\langle r\rangle\approx0.38629$ for the Poisson distribution~\cite{atas2013distribution}. 
In Fig.~\ref{fig:gapratio}, we plot the average gap ratio for different number of qubits $N$ and disorder $W$, finding that it converges towards the theoretical expected values in the respective chaotic and localized phases.
Furthermore, we find that curves for different $N$ cross at single disorder $W_\text{c}\approx0.42$, which indicates the transition point in the limit $N \to \infty$.

\begin{figure}
\centering
\subfigimg[width=0.5\textwidth]{}{gapratioT.pdf}
\caption{Average gap ratio $\langle r \rangle$ against disorder $W$ of the 1D mixed-field Ising Hamiltonian for different qubit number $N$. Dashed vertical line is estimation of MBL transition at $W_\text{c}\approx 0.42$. Horizontal lines denote the average gap ratios $\langle r\rangle\approx0.53590$ for the Wigner-Dyson distribution, and $\langle r\rangle\approx0.38629$ for the Poissonian distribution.}
 \label{fig:gapratio}
\end{figure}

\section{Shadow tomography with projected ensembles}

\change{Here, we show how projected ensembles can be used for classical shadow tomography~\cite{huang2020predicting} for learning properties of unknown states. We largely follow the treatment by McGinley and Fava~\cite{mcginley2023shadow}, although the presentation is adapted for our purposes. 

First, let us give a high-level introduction to shadow tomography in general and our projected ensemble protocol for shadow tomography:

Shadow tomography measures many expectation values of an unknown state $\rho$ from randomized measurements and classical post-processing~\cite{huang2020predicting}. 
In the standard protocol, the unknown quantum state is scrambled with a random unitary, and then measured in the computational basis. This is repeated for many different random unitaries. Then, in classical post-processing, one applies the inverse of the unitary on the measured outcomes. The resulting states are used to construct a classical representation of the unknown state to estimate expectation values.

However, the standard protocol requires the implementation of random unitaries drawn from a reference distribution (e.g., a unitary $3$-design), which can be difficult to implement in quantum simulators.
Instead, Refs.~\cite{tran2023measuring,mcginley2023shadow} proposed to use projected ensembles to generate the random dynamics, not necessarily unitary. Here, the unknown state is subsystem $A$, while one attaches (fixed) bath qubits as subsystem $B$. Then, a fixed unitary is applied to the full state, and all the qubits are measured. In the context of quantum simulators, this can be naturally realized by a quench dynamics $U = \exp(-iHT)$ for a fixed Hamiltonian $H$ and quench time $T$. Then, in the classical post-processing, one applies the inverse of the unitary on the post-measurement state, followed by an inverse map to construct an estimator of the target state $\rho$, and uses this to estimate expectation values. Notably, the measurements on the bath qubits supply the randomness needed for the protocol.

Our work now improves the projected measurement protocol by randomizing the initial state of the bath qubits. We show that our protocol gives an exponential improvement of the bias error of the shadow tomography protocol, substantially increasing accuracy and reducing measurement cost.

Conceptually, the projected ensembles protocol replaces the step of sampling random unitaries in the original classical shadows protocol~\cite{huang2020predicting} with random quantum evolution on subsystem $A$ induced by measurements on the bath qubits in subsystem $B$. The random unitaries used in Ref.~\cite{huang2020predicting} are sampled from an exact $2$-design, resulting in an unbiased shadow estimator. On the other hand, the projected ensembles protocol produces an approximate $2$-design for a finite bath size $N_B$, and thus causes a bias error. As we have shown in the main text, injecting classical randomness with entropy $S_c$ causes the projected ensemble to converge exponentially to the Haar ensemble. Consequently, the bias error of the shadow estimation is exponentially reduced with $S_c$, illustrated in Fig. 3(a) of the main text.


We now give an in-depth explanation of our protocol.}
We initialize the system in $\rho \otimes \ket{x}\bra{x}$, where $\rho$ is an unknown $N_A$-qubit state, and $\ket{x}$ is an $N_B$-qubit computational basis state, sampled from the classical distribution $q(x)$, $x = \{0,1\}^{N_B}$, with entropy $S_c \leq N_B$. The system is evolved under a fixed unitary $U$. Measuring subsystem $B$ in the computational basis $\{\ket{z_B}\}$ yields the projected ensemble $\mathcal{E}$, which comprises the post-measurement states
\begin{equation}
    \frac{(I_A \otimes \bra{z_B})U(\rho \otimes \ket{x}\bra{x})U^\dag(I_A \otimes \ket{z_B})}{\Tr (I_A \otimes \bra{z_B})U(\rho \otimes \ket{x}\bra{x})U^\dag(I_A \otimes \ket{z_B})} \otimes \ket{z_B}\bra{z_B}
\end{equation}
with weights $\Tr (I_A \otimes \bra{z_B})U(\rho \otimes \ket{x}\bra{x})U^\dag(I_A \otimes \ket{z_B})$. The $N_A$ qubits in the projected  states are then measured in the computational basis $\{\ket{z_A}\}$. Thus, we end up with the states $\ket{z_A} \otimes \ket{z_B}$ with outcome probabilities
\begin{equation}
    p_x(z_A,z_B) = \Tr (\bra{z_A} \otimes \bra{z_B})U(\rho \otimes \ket{x}\bra{x})U^\dag(\ket{z_A}\otimes \ket{z_B}).
\label{eq:zazb_distr}
\end{equation}
Note that the sequential measurements on $B$ followed by $A$ is purely conceptual. Since the measurements on $A$ and $B$ commute, they can be performed simultaneously, which can be done in practice with a global readout of the $N_A + N_B$ qubits. 

The main task of the quantum processor is to sample the bit strings $(z_A, z_B)$ from the distribution~\eqref{eq:zazb_distr}. In the classical post-processing, we construct the shadow estimator for $\rho$,
\begin{equation}
    \hat{\rho}_x(z_A,z_B) = (d_A + 1) \hat{\chi}_x(z_A,z_B) - I_A,
\label{eq:shadow_estimator}
\end{equation}
corresponding to each measurement outcome $(z_A,z_B)$, and
\begin{equation}
    \hat{\chi}_x(z_A,z_B) = \frac{(I_A \otimes \bra{x})U^\dag \ket{z_A,z_B}\bra{z_A,z_B}U (I_A \otimes \ket{x})}{\Tr (I_A \otimes \bra{x})U^\dag \ket{z_A,z_B}\bra{z_A,z_B}U (I_A \otimes \ket{x})}.
\end{equation}
$I_A$ is the identity operator on the $N_A$ qubits. The estimator satisfies the normalization $\Tr (\hat{\rho}_x(z_A,z_B)) = 1$. This can be used to construct an estimator for an arbitrary observable $O$, which we denote as
\begin{equation}
    \hat{O} = \Tr(O \hat{\rho}_x(z_A,z_B)).
\end{equation}
The estimator $\hat{O}$ is averaged over $L$ measurement shots. The error of the estimation, $\delta O$, is defined as
\begin{equation}
    \delta O = |\hat{O} - \Tr(O \rho)|. 
\end{equation}
As $L \to \infty$, the statistical fluctuations from measurement shot noise vanish, and $\delta O$ converges to the bias error, which can be written explicitly as
\begin{equation}
    \text{Bias error} = \sum_x \sum_{z_A,z_B} q(x) p_x(z_A z_B) \Tr(O\hat{\rho}_x(z_A,z_B)) - \Tr(O\rho).
\end{equation}
The sign of the bias error indicates a systematic over- or under-estimation of the true observable value $\Tr(O\rho)$.   

Since $U$ is fixed, an analytical treatment of the bias error is non-trivial and depends on $U$. Instead, for analytical tractability, we consider the average behavior of the bias error when $U$ is a Haar random unitary,
\begin{equation}
    \Eset{U \sim \text{Haar}(d_A d_B)} (\text{Bias error}) = \sum_x q(x) \E_U \sum_{z_A,z_B} \sparens{p_x(z_A,z_B) \Tr(O \hat{\rho}_x(z_A,z_B))} - \Tr(O\rho).
\label{eq:haar_shadowbias_defn}
\end{equation}
The Haar average bias error can be analytically calculated to be zero.
\begin{proposition}[Average bias of shadow estimation with projected ensemble]
    The average bias of the shadow estimation, defined in Eq.~\eqref{eq:haar_shadowbias_defn}, is zero. 
\end{proposition}
\begin{proof}
    The main computation is to evaluate the Haar average in Eq.~\eqref{eq:haar_shadowbias_defn} over $U$, which is not immediately straightforward due to the appearance of $U$ in the denominator of $\hat{\rho}_x(z_A,z_B)$ in Eq.~\eqref{eq:shadow_estimator}. As used in the proof of Theorem 1, we invoke the replica trick here. Let us define
    \begin{equation}
    \begin{split}
        f_m(U) &= \sum_{z_A,z_B} p_x(z_A,z_B) \Tr\sparens{(I_A \otimes \bra{x})U^\dag \ket{z_A,z_B}\bra{z_A,z_B}U(O \otimes \ket{x})} \\&\times \Tr^m\sparens{(I_A \otimes \bra{x})U^\dag \ket{z_A,z_B}\bra{z_A,z_B}U(I_A \otimes \ket{x})}.
    \end{split}
    \end{equation}
where we treat $m$ as a positive integer. The Haar average of $f_m$ can be evaluated using Weingarten calculus, and we take $m \to -1$ at the end to get $\E_U \delta O = (d_A + 1)\sum_x q(x) \lim_{m \to -1} \E_U f_m(U) - \Tr (O) -\Tr(O \rho)$. Substituting the definition of $p_x(z_A,z_B)$ from Eq.~\eqref{eq:zazb_distr}, we obtain 
\begin{equation}
\begin{split}
    \Eset{U \sim \text{Haar}(d_A d_B)} f_m(U) &= \frac{(d_A d_B - 1)!}{(d_A d_B + m +1)!} \sum_{z_A,z_B} \sum_{\pi \in S_{m+2}} \Tr(\hat{\pi}_A (\rho \otimes O \otimes I_A^{\otimes m})) \\&= \frac{(d_A d_B)!}{(d_A d_B + m +1)!} \frac{(d_A + m+1)!}{(d_A +1)!} \sum_{\pi \in S_2} \Tr(\hat{\pi}_A(\rho \otimes O)) \\&= \frac{(d_A d_B)!}{(d_A d_B + m +1)!} \frac{(d_A + m+1)!}{(d_A +1)!} \sparens{\Tr (O) + \Tr(O\rho)}.
\end{split}
\end{equation}
Taking the limit $m \to -1$ yields
\begin{equation}
    \lim_{m \to -1} \Eset{U\sim \text{Haar}(d_A d_B)} f_m(U) = \frac{1}{d_A + 1}\sparens{\Tr O + \Tr(O\rho)}
\end{equation}
which implies that the bias error, on average, is zero.
\end{proof}

Although the average bias error is zero, this does not mean that the bias error is zero for a fixed unitary $U$, even if sampled from the Haar distribution. A detailed analysis of the bias error for a fixed $U$ is beyond the scope of this work. Here, we provide an intuitive account for the bias error.

In defining the shadow estimator in Eq.~\eqref{eq:shadow_estimator}, we are essentially assuming that the ensemble of $\hat{\chi}_x(z_A,z_B)$ (or the tomographic ensemble, in the language of Ref.~\cite{mcginley2023shadow}) forms an exact state $2$-design. While this is true in the limit $N_B \to \infty$, one only obtains an approximate $2$-design for a finite $N_B$. From our analysis of the projected ensemble, we expect the distance to the Haar ensemble to scale as
\begin{equation}
    \Delta_\alpha^{(2)} \sim \frac{1}{2^{(S_c + N_B)/2}}
\end{equation}
with the number of bath qubits $N_B$ and the entropy $S_c$ of the classical distribution $q(x)$, for any Schatten index $\alpha$ in Definition~\ref{defn:approxdesign}, up to factors of $D_k$ which does not involve $N_B$ and $S_c$. This is exactly the observed scaling behavior in Fig. 3(a) of the main text, where we plot the absolute value of the bias error, $\delta O (L \to \infty)$. Thus, by injecting classical randomness into the shadow tomography protocol, we can achieve an exponentially smaller bias error with a fixed bath size $N_B$.

\subsection{Patched quench}

\begin{figure}
\centering

\subfigimg[width=0.37\textwidth]{a}{PatchSketch.pdf}
\subfigimg[width=0.4\textwidth]{b}{errorEnergyShadowShdN10B0r100m6S5000M1t2e0s41.pdf}
\caption{
\idg{a} Setup for classical shadow tomography with projected ensembles and patch-wise evolution. Small circles represent qubits arranged in a $N_\text{A}\times(N_\text{B}+1)$ grid, vertical ellipses represent quantum states, and horizontal square boxes represent on which qubits the unitaries are applied. The goal is to perform classical shadow tomography on the unknown $N_\text{A}$-qubit state $\ket{\psi_\text{g}}$ on the left column. $N_\text{B} \times N_\text{A}$ qubits are prepared in computational basis states $\ket{x} \equiv \ket{x_1}\otimes\dots\otimes\ket{x_{N_\text{B}}}$. Then, an $(N_\text{B}+1)$-qubit unitary $V$ is applied in parallel row-wise on all qubits. Finally, the full system is measured in the computational basis, and one performs the usual post-processing for classical shadows. Since the size of $V$ can be chosen independently of $N_\text{A}$, this is computationally efficient. Classical entropy $S_{\text{c}}$ is introduced by randomizing the computational basis states $\ket{x_i}$ between each round of measurements, which greatly reduces the estimation error.
\idg{b} Error $\delta E/N_A$ of the estimated energy density $E/N_A$ of ground state $\ket{\psi_\text{g}}$ of the mixed-field Ising Hamiltonian $H_{N_A}$~\eqref{eq:Hising_patch} using the `patched quench' shadow tomography protocol. We evolve with unitary $V=\exp(-i H_{N_B+1} T)$. The error is plotted against the number of system qubits $N_A$ for fixed ancilla state ($S_{\text{c}}=0$) and randomly chosen computational basis ancilla states with classical randomness $S_{\text{c}}/N_A=N_B$. Black dashed line is the corresponding result for the shadow tomography protocol of Ref.~\cite{huang2020predicting} with Haar random single-qubit unitaries as reference. We have $N_B=6$ bath qubits, $L=5000$ measurement shots, and quench time $JT=100$.}
 \label{fig:Patchshadow}
\end{figure}

In the classical post-processing given by Eq.~\eqref{eq:shadow_estimator}, we require an efficient classical simulation of the inverse evolution $U^\dagger$~\cite{huang2020predicting}. In our case, $U$ is typically fixed by the many-body Hamiltonian dynamics which is in general exponentially hard to simulate classically. To resolve this issue, the `patched quench' approach was proposed to apply the evolution $U=V^{\otimes{q}}$ on $q$ disjoint patches of qubits, using smaller unitaries $V$~\cite{tran2023measuring}. Here, the size of unitary $V$ is chosen such that it can be inverted classically.

The scheme is sketched in Fig.~\ref{fig:Patchshadow}(a).
We prepare a two-dimensional grid of size $N_A \times (N_B+1)$ qubits with the initial state $\ket{\psi_\text{g}}\otimes \ket{x}$ where the ancilla of size $N_A N_B$ qubits is prepared in a computational basis state $\ket{x}$, and $\ket{\psi_g}$ is the unknown state to be learned. Here, we arrange $\ket{\psi_\text{g}}$ along the first column of the 2D grid. Then, we evolve the system and ancilla together via the unitary $U=V^{\otimes{N_A}}$ with $N_A$ disjoint patches. Each patch unitary $V=\exp(-i H T)$ acts on a row of the grid, i.e., on one system qubit and $N_B$ ancilla qubits. 

After evolution, all the qubits are measured in the computational basis. In the classical post-processing, the inverse evolution $U^\dagger=\exp(i H T)^{\otimes{N_A}}$ is applied to construct $N_A$ individual single-qubit shadow estimators of the form in Eq.~\eqref{eq:shadow_estimator} (with $d_A = 2$, since each patch only contains one sys). We note that for a constant $N_B$, one can efficiently invert $V$ for the shadow tomography as it acts only on $N_B+1$ qubits. This remains efficient even when scaling up $N_A$. From our numerical results shown in the main text, the bias error decreases exponentially with $N_B$. This can be further reduced exponentially by injecting classical randomness (for example, sampling the ancilla state $\ket{x}$ from a classical distribution). 

Essentially, the `patched quench' approach mimics the original classical shadows protocol in Ref.~\cite{huang2020predicting} of applying random single-qubit unitaries to the system, which is now replaced with a random evolution induced by bath measurements. This approach is thus useful for efficiently learning many non-commuting low-weight observables, especially in many quantum simulators where one cannot easily access or control single-qubit dynamics.

In Fig.~\ref{fig:Patchshadow}(a), we illustrate the scheme where each patch consists of one system qubit and $N_B$ ancilla qubits. This scheme can be easily generalized to include more than one system qubit in each patch, or have patches that act on varying number of qubits. The classical simulation cost grows polynomially with the number of patches and exponentially with the number of qubits in the patch.

\subsection{Numerical demonstration: Estimating ground state energy density}
To implement $U=V^{\otimes{N_A}}=\exp(-i H T)^{\otimes{N_A}}$, one needs to evolve the qubits with the Hamiltonian $H$ separately along each row of the two-dimensional grid. This can be physically achieved using two-dimensional Ising simulators where the rows can be decoupled. In cold atom simulators, the decoupling of the qubits can be realized by adding a tunnel barrier between the rows. In other architectures such as trapped ions and tweezer arrays of  neutral atoms, the decoupling can be done by moving the traps apart, at a timescale that is sufficiently faster than the decoherence time of the qubits.

Here, for illustration purposes, we consider a quantum simulator where qubits on each row interact via the 1D mixed-field Ising Hamiltonian $H_{N_B+1}$, where 
\begin{equation}
\label{eq:Hising_patch}
    H_N \equiv \sum_{i=1}^{N} (h_x X_i + h_y Y_i + J X_i X_{i+1}),
\end{equation}
with parameters $\{h_x, h_y, J\} = \{0.8090,0.9045,1\}$ and open boundary conditions.
The initial state $\ket{\psi_g}$ is prepared in the ground state of $H_{N_A}$.
Our goal is to efficiently estimate the ground state energy density $E/N_A=\bra{\psi_\text{g}}H_{N_A}\ket{\psi_\text{g}}/N_A$. 



Applying the `patch quench' protocol, we obtain an estimation for the energy density $\hat{E}/N_A$ (averaged over $L = 5000$ measurement shots) from $N_A$ individual shadow estimators, with an error $\delta E/N_A= \vert \hat{E} - E\vert/N_A$. 
In each run, we either keep ancilla state $\ket{x}$ fixed (i.e. classical randomness $S_{\text{c}}=0$) or randomize the computational basis state ($S_{\text{c}}/N_A=N_B$).
In Fig.~\ref{fig:Patchshadow}(b), we plot $\delta E /N_A$ against $N_A$. We find adding classical randomness into the initialization reduces the error $\delta E / N_A$ by nearly an order of magnitude. The reduction in error is nearly independent of $N_A$, indicating that our protocol can achieve low estimation errors in a scalable manner. Note that our protocol is  computationally efficient while only requiring access to the native many-body Hamiltonian of the analog quantum simulator.

\subsection{Additional numerical results on bias error of shadow tomography}
In Fig.~\ref{fig:shadowstd}, we show the standard deviation of the bias error $\delta O$ in our shadow tomography protocol of Fig. 3(a) of the main text. We find that the standard deviation is on the same order as the average value of the bias error.
\begin{figure}
\centering
\subfigimg[width=0.4\textwidth]{}{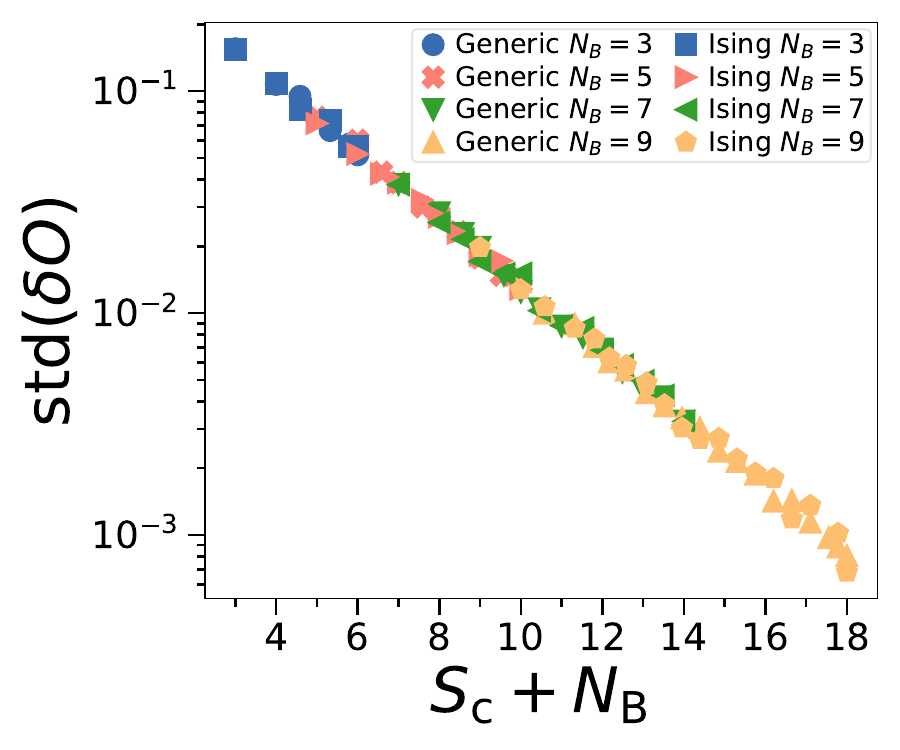}
\caption{\revA{Standard deviation} of bias error $\delta O$ (with $L \to \infty$ measurement shots) against classical entropy $S_{\text{c}}$ plus number of bath qubits $N_\text{B}$. For `Generic', we randomize the initial state of the $N_\text{B}$ bath qubits and evolve with a fixed Haar random unitary. 
For `Ising', we evolve with Eq. (5) (in the main text) for $JT=100$,  where the initial state is fixed and the disorder of the Hamiltonian is randomized for $2^{S_c}$ realizations with $W/J=0.5$. 
We show $N_A=1$, and take standard deviation over 100 randomly chosen instances of unitaries.  }
 \label{fig:shadowstd}
\end{figure}

\section{Deep thermalization from generic Hamiltonian evolution}
In the main text, the analytical result is derived for the evolution under a (fixed) unitary drawn from the Haar measure, which we show agrees well with numerical simulations of the MFIM. Here, we provide analytical evidence that deep thermalization occurs for generic Hamiltonian. To this end, we consider an ensemble of Hamiltonians $\mathcal{E}_H$ which is invariant under unitary conjugation. This includes, as a special case, the Gaussian Unitary Ensemble (GUE) which is widely studied for quantum chaos.  We only sketch the proof idea below, while the detailed rigorous proof is beyond the scope of this paper and deferred to a forthcoming manuscript.

Assuming the probability for any bath outcome $z$ is approximated as $p(z) \approx 1/d_B$ (this will be addressed below), the frame potential of the projected ensemble is approximately
\begin{equation}
    F^{(k)} \approx \frac{1}{d_B} + d_B^{2k-2} \sum_{z_1 \neq z_2} |\braket{0|e^{iHT}(I_A \otimes \ket{z_1}\bra{z_2})e^{-iHT}|0}|^{2k},
\end{equation}
where $T$ is the evolution time and $\ket{0}$ is the initial state. Since $H = U^\dag \Lambda U$ is diagonalized by a Haar random unitary $U$ with $\Lambda$ being the diagonal matrix containing the eigenenergies, averaging over $\mathcal{E}_H$ is equivalent to averaging over $U$ and the spectrum $\Lambda$ separately. This gives
\begin{equation}
    \E_{\mathcal{E}_H} F^{(k)} = \E_{\Lambda,U} F^{(k)} \approx \frac{1}{d_B} + d_B^{2k-2} \sum_{z_1\neq z_2} \E_{\Lambda,U} |\braket{0|U^\dag e^{i\Lambda T} U (I_A \otimes \ket{z_1}\bra{z_2}) U^\dag e^{-i\Lambda T} U|0}|^{2k}.
\label{eq:hamUE_avgF}
\end{equation}
Evaluating the Haar average, and keeping the leading order term as $d_A, d_B \to \infty$ yields
\begin{equation}
    \E_{\mathcal{E}_H} F^{(k)} \approx \frac{1}{d_B} + \frac{k!}{d_A^k} \approx \frac{1}{d_B} + F_{\text{Haar}}^{(k)},
\label{eq:hamUE_avgF_simplified}
\end{equation}
which matches the result in Theorem 1 of the main text for $S_c = 0$. The calculation can be extended to the case of $S_c > 0$, using a similar approach as in Sec.~\ref{sec:sm_derivation}. In deriving Eq.~\eqref{eq:hamUE_avgF_simplified} from Eq.~\eqref{eq:hamUE_avgF}, it might seem like the average frame potential is independent of $\Lambda$ and $T$, which is not true. Going back to our initial assumption, the measurement probability can be evaluated as
\begin{equation}
    \E_{\mathcal{E}_H} p(z) = \E_{\Lambda,U} \braket{0|U^\dag e^{i\Lambda T} U (I_A \otimes \ket{z}\bra{z})U^\dag e^{-i\Lambda T} U|0} \approx \frac{1}{d_A^2 d_B^2} \sparens{d_A^2 d_B + \mathcal{O}\left(R_2(T) \times |\braket{z|0}|^2\right)},
\end{equation}
where $R_2(T) \equiv \E_\Lambda \sparens{|\text{Tr} (e^{i\Lambda T})|^2}$ is the (2-point) spectral form factor. For sufficiently large $T$, $R_2(T)$ typically decays far below its initial value of $d_A^2 d_B^2$. In the case of GUE, it has been shown that $R_2(T) \sim \sqrt{d_A d_B}$ at the dip time and $\sim d_A d_B$ at late times~\cite{cotler2017chaos}. In either case, $\E_{\mathcal{E}_H} p(z) \approx 1/d_B$ to leading order, justifying the assumption. The statistical fluctuations of $p(z)$ can be controlled by evaluating the variance of $p(z)$, which involves the $4$-point spectral form factor $R_4(T) \equiv \E_\Lambda \sparens{|\text{Tr} (e^{i\Lambda T})|^4}$. In other words, one can relate the emergence of deep thermalization to dynamical signatures of chaos.

Although the `typical' Hamiltonian here is rather unphysical and describes essentially infinite-range interactions, such an approach based on random matrix theory is well established in explaining the statistical properties for quantum thermalization in local chaotic Hamiltonians~\cite{dAlessio2016quantum} without time-reversal symmetry. In particular, this is a reasonable description of the late-time behavior (i.e., past the scrambling time) of chaotic Hamiltonians, where the locality of the interactions are not significant~\cite{cotler2017chaos}.

Additionally, we also expect our conclusions to also hold for generic Hamiltonians obeying time-reversal symmetry, in which the Haar random unitary $U$ should be replaced by a Haar random orthogonal matrix $O$. The projected ensemble in this case is expected to approximate the ensemble of real-valued Haar random states. As shown in a recent work~\cite{schatzki2024random}, such an ensemble is statistically indistinguishable from complex-valued Haar random states (which we study) for large $N_A$, with the trace distance between the two ensembles vanishing exponentially as $\sim k^2/2^{N_A}$.

\bibliography{bib}